\immediate\write18{bibtex lmcs}
\documentclass{LMCS}

\def\dOi{11(3:2)2015}
\lmcsheading%
{\dOi}
{1--23}
{}
{}
{Jan.~13, 2013}
{Aug.~\phantom07, 2015}
{}

\ACMCCS{[{\bf Theory of computation}]: Formal languages and automata
  theory---Formalisms---Algebraic language theory; Semantics and
  reasoning---Program semantics---Operational
  semantics\,/\,Denotational semantics} 


\usepackage{amsmath,amsfonts,latexsym,amssymb}
\usepackage[all]{xy}
\usepackage{comment}
\usepackage{xcolor}
\usepackage{hyperref}
\usepackage{bbm}
\theoremstyle{plain}
\newtheorem{lemma}[thm]{Lemma}
\newtheorem{proposition}[thm]{Proposition}
\newtheorem{theorem}[thm]{Theorem}
\newtheorem{corollary}[thm]{Corollary}
\theoremstyle{definition}
\newtheorem{definition}[thm]{Definition}
\newtheorem{example}[thm]{Example}
\newtheorem{remark}[thm]{Remark}
\newtheorem{assumption}[thm]{Assumption}

\theoremstyle{plain}

\renewcommand{\varepsilon}{\epsilon}

\newcommand{\defiEnd}{\hspace{\stretch{1}}\ensuremath{\triangleleft}}

\newcommand{\cat}[1]{\mathsf{#1}}
\newcommand{\bb}[1]{\mathbb{#1}}

\newcommand{\C}{\cat{C}} % our base category
\newcommand{\A}{\cat{A}}
\newcommand{\Set}{\cat{Set}}   % category of sets and functions
\newcommand{\Pos}{\cat{Pos}}   % category of sets and functions
\newcommand{\Coalg}{\mathsf{Coalg}}
\newcommand{\Bialg}{\mathsf{Bialg}}
\newcommand{\Alg}[1]{\cat{Alg}(#1)} % Eilenberg-Moore category
\newcommand{\Algeq}[2]{\cat{Alg}(#1,#2)} % EM with equations

\newcommand{\ext}[2]{#1^{#2}}

\newcommand{\Pow}{\mathcal{P}}
\newcommand{\Id}{\mathrm{Id}}
\newcommand{\id}{\mathrm{id}}

\newcommand{\E}{\mathcal{E}} % equations
\newcommand{\T}{\mathcal{T}}
\newcommand{\K}{\mathcal{K}}
\newcommand{\eql}{l}
\newcommand{\eqr}{r}

\newcommand{\bbN}{\mathbb{N}}
\newcommand{\bbR}{\mathbb{R}}
\newcommand{\bbZ}{\mathbb{Z}}
\newcommand{\twobb}{\mathbbm{2}}

\renewcommand{\to}{\rightarrow}
\newcommand{\To}{\Rightarrow}

\newcommand{\Eq}[1]{\equiv_{#1}}  % \E{X} = eqrel on terms with variables in X
\newcommand{\Rel}[1]{\overline{#1}}  % relation lifting

\newcommand{\quotE}[1]{{#1}_\E}

\newcommand{\abstractgoes}[2]{%
\setbox0=\hbox{\ ${\scriptstyle#2}$\ }
\ifdim\wd0<12pt\wd0=12pt\fi
\mathrel{\stackrel{#2}{\rule[2.2pt]{\wd0}{0.6pt}}\mkern-16mu{#1}}
}

\newcommand{\sse}{\subseteq}

\newcommand{\tup}[1]{\langle #1 \rangle}
\newcommand{\ol}[1]{\overline{#1}}  %overline
\newcommand{\sig}{\sigma} % shorthand
\newcommand{\X}{\mathtt{X}} %  constant X = (0,1,0,0,..)
\newcommand{\cns}[1]{\mathtt{#1}} %  \cns{r} = constant r = (r,0,0,0,..)
\newcommand{\Sig}{\Sigma} % shorthand
\begin{document}

%\title[short title]{How to use lmcs.cls}

%\author[Author1 et al.]{Author 1}	%required
%\address{address 1}	%required
%\email{author1@email1}  %optional
%%\thanks{thanks 1, optional.}	%optional

%\author[]{Author 2}	%optional
%\address{address2; addresses should be duplicated when authors share an affiliation}	%optional
%\email{author2@email2; ditto for email addresses}  %optional
%\thanks{thanks 2, optional.}	%optional

%\author[]{Author 3}	%optional
%\address{address 3}	%optional
%\email{author3@email3}  %optional
%\thanks{thanks 3, optional.}	%optional

\title[Presenting Distributive Laws]{Presenting Distributive Laws}

\author[M.~M.~Bonsangue]{Marcello M. Bonsangue\rsuper a}	%required
\address{{\lsuper{a,d}}LIACS -- Leiden University, Netherlands, \and Formal Methods -- Centrum Wiskunde \& Informatica, Amsterdam, Netherlands}	 %required
\email{marcello@liacs.nl, j.c.rot@liacs.leidenuniv.nl}
%\thanks{thanks 1, optional.}	%optional

\author[H.~H.]{Helle H. Hansen\rsuper b}	%optional
\address{{\lsuper b}ESS/TBM -- Delft University of Technology, Netherlands, \and Formal Methods -- Centrum Wiskunde \& Informatica, Amsterdam, Netherlands}	%optional
\email{h.h.hansen@tudelft.nl}  %optional
\thanks{{\lsuper b}The research of Helle H. Hansen has been funded by the Netherlands Organisation for
Scientific Research (NWO), Veni project number 639.021.231.}	%optional

\author[A.~Kurz]{Alexander Kurz\rsuper c}	%optional
\address{{\lsuper c}Dep.~of Computer Science -- University of Leicester, United Kingdom}	%optional
\email{ak155@le.ac.uk}
%\thanks{thanks 3, optional.}	%optional

\author[J.~Rot]{Jurriaan~Rot\rsuper d}	%optional
\address{\vspace{-18 pt}}	 %optional
%\email{author3@email3}  %optional
\thanks{{\lsuper d}The research of Jurriaan~Rot has been funded by the Netherlands Organisation for
Scientific Research (NWO), CoRE project, dossier number: 612.063.920.}	%optional

%\author{Marcello M. Bonsangue\inst{1,3} \and Helle H. Hansen \inst{2,3}\fnmsep\thanks{The research of this author has been funded
% by the Netherlands Organisation for Scientific Research (NWO), Veni project number 639.021.231.} \and
% Alexander Kurz\inst{4} \and Jurriaan~Rot\inst{1,3}\fnmsep\thanks{The research of this author has been funded
% by the Netherlands Organisation for Scientific Research (NWO), CoRE project, dossier number: 612.063.920.} }
%\authorrunning{M.M. Bonsangue \and H.H. Hansen \and A. Kurz \and J. Rot}
%\institute{
%LIACS -- Leiden University, Netherlands
%\and
%ICIS/IS -- Radboud University Nijmegen, Netherlands
%\and
%Formal Methods -- Centrum Wiskunde \& Informatica, Amsterdam, Netherlands
%\and
%Dep.~of Computer Science -- University of Leicester, United Kingdom
%}

%% etc.

%% required for running head on odd and even pages, use suitable
%% abbreviations in case of long titles and many authors:

%% mandatory lists of keywords and classifications:
%\titlecomment{OPTIONAL comment concerning the title, \eg, if a variant
%or an extended abstract of the paper has appeared elsewehere}
%%%%%%%%%%%%%%%%%%%%%%%%%%%%%%%%%%%%%%%%%%%%%%%%%%%%%%%%%%%%%%%%%%%%%%%%%%%

%% the abstract has to PRECEED the command \maketitle:
%% be sure not to issue the \maketitle command twice!

\begin{abstract}
Distributive laws of a monad $\T$ over a functor $F$
are categorical tools for specifying algebra-coalgebra interaction.
They proved to be important for
solving systems of corecursive equations, for the
specification of well-behaved structural operational semantics and, more recently, also
for enhancements of the bisimulation proof method.
If $\T$ is a free monad, then such distributive laws correspond to simple natural
transformations. However, when $\T$ is not free it can be rather
difficult to prove the defining axioms of a distributive law. In this paper we describe
how to obtain a distributive law for a monad with an equational presentation from a
distributive law for the underlying free monad.
We apply this result to show the equivalence between two different representations of context-free languages.
\end{abstract}

\maketitle

%------------------------------------------------------------------
% Use only during writing:

%\hide{ \tableofcontents }

%------------------------------------------------------------------

\section{Introduction}
The combination of algebraic structure and observable behaviour
is fundamental in computer science.
Examples include the operational models of structural
operational semantics~\cite{Aceto:SOS-HB},
denotational models of programming languages~\cite{winskel},
finite stream circuits~\cite{Milius10},
linear and context-free systems of
behavioural differential equations~\cite{Rut03:TCS-bde,WBR:CF-CALCO},
and many types of automata such as nondeterministic and
weighted automata~\cite{SilvaBBR10}.

In the categorical treatment of these examples,
the algebraic structure is encoded by a monad $\T = \tup{T,\eta,\mu}$, and
the system behaviour by coalgebras for a functor $F$.
Often it is desirable that the algebraic and coalgebraic structure
are compatible in some way. A general approach to specifying such
algebra-coalgebra interaction is via a distributive law.
There are several advantages of this structured approach.
A distributive law $\lambda$ of the monad $\T$ over $F$
induces a $\T$-algebra on the final $F$-coalgebra of behaviours,
yields solutions to corecursive equations $\phi\colon X \to FTX$ ~\cite{Bartels:PhD},
and ensures that bisimulation is a congruence~\cite{TuriPlotkin:LICS-GSOS}.
Moreover, it yields the soundness of techniques such as
bisimulation-up-to-context~\cite{Bartels:PhD} and extensions
thereof~\cite{RBBRS13,RBR13}.

Describing a distributive law explicitly and proving that it is one can,
however, be rather complicated. Therefore, general methods
for constructing distributive laws from simpler ingredients are very useful.
An important example of this is given by
abstract GSOS~\cite{TuriPlotkin:LICS-GSOS,Bartels:PhD,LPW2004:cat-sos}
where distributive laws of a free monad $\T$ over a (copointed) functor $F$
are shown to correspond to plain natural transformations,
called \emph{abstract GSOS-rules} as
they can be seen as specification formats.
In \cite{HK:pointwise} it was shown how
an abstract GSOS-rule for a free monad $\T$ and functor $F$
can be lifted to one for the functor $F(-)^A$
which describes $F$-systems with input in $A$.
Another method which works for all monads $\T$, but only for certain polynomial
behaviour functors $F$, produces a distributive law inducing
a ``pointwise lifting'' of $\T$-algebra structure to $F$-behaviours,
cf.~\cite{Jacobs:bialg-dfa-regex,Jacobs06,SilvaBBR10}.

But many examples do not fit into the abovementioned settings.
An important motivating example for this paper is that of context-free grammars,
where sequential composition is not a pointwise operation and
whose formal semantics satisfies the axioms of idempotent semirings,
i.e., the algebraic structure is not free.
More generally, one may be interested in a monad arising from a free one
by adding equations which one knows to hold in the final coalgebra.
%, without having a particular concrete monad in mind.

The main contribution of this paper is to give a general approach
for constructing a distributive law $\lambda'$
for a monad $\T'$ with an equational presentation,
from a distributive law $\lambda$ for the underlying free monad $\T$.
We have no constraints on the behaviour functor $F$.
This $\lambda'$ is obtained as a certain quotient of $\lambda$
by the equations $\E$ of $\T'$, hence we say that
\emph{$\lambda'$ is presented by a $\lambda$ for the free monad and the equations $\E$}.
We show that such quotients exist precisely when the distributive law \emph{preserves the equations $\E$},
which roughly means that congruences generated by the equations are
bisimulations.
We also discuss how these quotients of distributive laws give rise to
quotients of bialgebras, thereby giving a concrete operational interpretation,
and a correspondence between solutions to corecursive equations with
and without equations.
As an illustration and application of our theory, we show the existence of a distributive law
of the monad for idempotent semirings over the deterministic automata functor.
This result yields the equivalence betweeen the Greibach normal form representation of context-free
languages and the coalgebraic representation via context-free expressions given in~\cite{WBR:CF-CALCO}.

\textit{Outline.} In Section~\ref{sec:quot-monad}, we recall the notions of monads and algebras, and
give a concrete description of monad quotients.
In Section~\ref{sec:dls}, we recall distributive laws and their application to solving
systems of equations. Then, in Section~\ref{sec:quotients-of-dls}, we prove our main results
on quotients of distributive laws. In Section~\ref{sec:quotients-of-bialgebras},
we show that such quotients give rise to quotients of bialgebras. Finally, in Section~\ref{sec:conc},
we discuss related work, and provide some directions for future work.

This paper is an extended version of~\cite{BHKR13}. It contains all proofs and generalises the main
results of~\cite{BHKR13} from monads on the category of $\Set$ to monads on arbitrary categories
(with the appropriate structure in the category of algebras). Furthermore, this paper includes
the treatment of distributive laws of a monad over a comonad, allowing more general specification
formats that involve look-ahead in the premises.

\textbf{Acknowledgements.}
We thank Neil Ghani, Bart Jacobs, Jan Rutten, Ji\v{r}\'{i}  Velebil and Joost Winter for helpful discussions
and suggestions. We are indebted to the referees for their constructive
comments, which greatly helped us improving the paper.

%------------------------------------------------------------------

\section{Monads, Algebras and Equations}\label{sec:quot-monad}

We start by recalling some basic definitions on monads, algebras, term equations
and congruences (see, e.g., ~\cite{BW05,ARV11} for a
detailed introduction). We will then proceed to give a concrete description of the
quotient monad arising from a free monad and a set of equations.
%In this section we consider only monads on $\Set$.

Let $\C$ be a category.
A \emph{monad} is a triple $\T = \tup{T, \eta, \mu}$ where
$T$ is an endofunctor on $\C$, and $\eta \colon \Id \Rightarrow T$ and $\mu \colon TT \Rightarrow T$ are natural transformations such
that $\mu \circ T \eta = \id = \mu \circ \eta_T$ and $\mu \circ \mu_T = \mu \circ T \mu$.
A \emph{$\T$-algebra} is a pair $\tup{A, \alpha}$ where $A$ is a $\C$-object and $\alpha \colon TA \rightarrow A$ is an arrow such that
$\alpha \circ \eta_A = \id$ and $\alpha \circ \mu_A = \alpha \circ T\alpha$.
A \emph{($\T$-algebra) homomorphism} from $\tup{A, \alpha}$ to $\tup{B, \beta}$
is an arrow $f \colon A \rightarrow B$ such that $f \circ \alpha = \beta \circ Tf$.
The \emph{free $\T$-algebra} over a $\C$-object $X$ is $\tup{TX,\mu_X}$. Given any $\T$-algebra $\tup{A, \alpha}$
and any arrow $f \colon X \rightarrow A$,
there is a unique algebra homomorphism $f^\sharp \colon TX \rightarrow A$ such that
$f^\sharp \circ \eta_X = f$, given by $f^\sharp = \alpha \circ Tf$.
We denote the category of $\T$-algebras and their homomorphisms by $\Alg{\T}$, and the associated forgetful functor
by $U \colon \Alg{\T} \rightarrow \C$.

Let $\tup{T,\eta,\mu}$ and $\tup{K,\theta,\nu}$ be monads. A \emph{morphism of monads} is a natural transformation $\sigma \colon T \Rightarrow K$ such
that the following diagram commutes:
\begin{equation}\label{eq:monadmap}
\xymatrix{
\Id \ar[r]^-{\eta} \ar[dr]_-{\theta}
 & T \ar[d]^-{\sigma}
 & TT \ar[l]_\mu \ar[d]^{\sigma \sigma}
\\
 & K
 & KK \ar[l]_{\nu}
}
\end{equation}
where $\sigma\sigma = K\sigma \circ \sigma_T = \sigma_K \circ T \sigma$.

Assume we are given a monad $\T = \tup{T, \eta, \mu}$ on some category $\C$.
We define \emph{$\T$-equations} as a 3-tuple $\E = \tup{E, \eql, \eqr}$ where $E$ is an endofunctor on $\C$ and
$\eql, \eqr \colon E \Rightarrow T$ are natural transformations.
The intuition is that $E$ models
the arity of the equations, and $l$ and $r$ give the left and right-hand side, respectively. This
is illustrated below by an example.
%In the remainder of this paper, in examples where $\T$ is a $\Set$ monad, for a clear presentation we will write equations as elements of $TV \times TV$ for some fixed set
%of variables $V$.

\begin{example}\label{ex:equations}
  Consider the $\Set$ functor $\Sigma X = X \times X + 1$, modelling a binary operation and a constant, which we call $+$ and $0$ respectively. The (underlying functor of the) free monad
  $T_\Sigma$ for $\Sigma$ sends a set $X$ to the terms over $X$ built from $+$ and $0$. The equations $x+0 = x$, $x+y = y+x$ and $(x+y)+z = x + (y + z)$ can
  be modelled as follows. The functor $E$ is defined as $EX = X + (X \times X) + (X \times X \times X)$. The natural transformations $\eql,\eqr \colon E \To T_\Sigma$ are
  given by $\eql_X(x) = x+0$ and $\eqr_X(x) = x$ for all $x \in X$; $\eql_X(x,y)=x+y$ and $\eqr_X(x,y) = y+x$ for all $(x,y) \in X \times X$;
 $\eql_X(x,y,z) = x+(y+z)$ and $\eqr_X(x,y,z) = (x+y)+z$ for all $(x,y,z) \in X \times X \times X$.
\end{example}

Throughout this paper we will need assumptions on $\C$, $\T$, and $\E$. For this section we only need the following.

\begin{assumption}\label{assumps}
%\marginpar{\blue{Assump's on $\C$?}}
We assume that $\T$ is a monad on $\C$, and $E\colon \C \to \C$ is a functor such that:
\begin{enumerate}
\item $\Alg{\T}$ has coequalizers.
%\item The forgetful functor $U \colon \Alg{\T} \to \C$ preserves regular epis.
\item $U$ and $TU$ map regular epis in $\Alg{\T}$ to epis in $\C$.
%\item $E$ preserves epis.
\item $EU$ maps regular epis in $\Alg{\T}$ to epis in $\C$.

\end{enumerate}
The first condition is needed to construct quotients of free algebras modulo equations.
%For example,
%if $\T$ is a monad on $\Set$ then $\Alg{\T}$ is cocomplete, and therefore has all coequalizers.
%A similar results holds also for other categories than $\Set$: if $\C$ is complete and cocomplete
%and $T$ preserves filtered colimits (i.e. $T$ is finitary), then $\Alg{\T}$  is complete and
%cocomplete as well~\cite{BW05}.
%
The second condition relates quotients of algebras (regular epis) with quotients in the base category (epis). It is satisfied if either $U$ preserves regular epis or if $U$ preserves epis.
%This
%condition implies that $T$ preserves regular epis, but it is clearly weaker than
%$U$ preserving coequalizers: if $s$ is the coequalizer of $g$ and $h$ in $\Alg{\T}$ (i.e.\ $s$ is a
%regular epi), then we can only conclude that $Uq$ is the coequaliser of a pair of maps in $\C$,
%but not necessarily of $Ug$ and $Uh$. In $\Set$, and more generally, in every topos $\C$,
%every epimorphism is regular. If $\T$ is a monad on $\Set$ (or on a topos $\C$) then
%regular epis in $\Alg{\T}$ are exactly those morphisms whose underlying morphism is an
%epimorphism $\Set$ (or, more generally, in a topos $\C$). In this case $U$ always preserves
%regular epis. \begin{ak}is it true for any topos? The argument in Set goes by saying that epis in Set split, hence are preserved by any $T$ ... but this is not true in toposes anymore\end{ak}
%
Finally, the last condition is satisfied if  $E$ preserves
epimorphisms in $\C$.

%In $\Set$ every functor preserves epimorphisms.
%More generally, if $E$
%preserves finite colimits then it preserves epimorphisms.
%(Concretely, we need that if $s$ is a coequalizer in $\Alg{\T}$,
%then $E(q)$ and $T(q) = UF(q)$ are epi in $\C$, where we consider $s$ as an arrow in $\C$.
%We also need that $Uq$ is epi.)
%For example, this holds for $\C=\Set$, any finitary monad $\T$ on $\Set$,
%and any $E\colon \Set \to\Set$.
%\mycomment{Alexander, any other interesting (non)examples we can mention? Monads on Posets?}
\end{assumption}

\begin{example}\label{exle:assumptions}\
\begin{enumerate}
\item  If $\C=\Set$ the conditions are satisfied for any monad $\T$ and endofunctor $E$.
\item If $\C$ is the category $\Pos$ of posets with monotone maps  the second condition may fail, see Example~6 in \cite{Bloom-Wright}, which can be adapted to show that the monad induced by the adjunction $\twobb\times(-)\dashv(-)^\twobb\colon\Pos\to\Pos$ does map some regular epis to non-epis.
\item\label{exle:assumptions:3} If $\C$ is abelian groups and $T$ is the identity, the well-known fact that torsion-free abelian groups are not monadic over $\Set$ \cite{Borceux:HB-vol2} can be adapted to give an example of equations $E$ that fail condition 3. Let $\bbZ_n$ denote the group of integers with addition modulo $n$. Define $E:\C\to\C$ as the coproduct of abelian groups $E(A)=\coprod_{n\in\bbN}\C(\bbZ_n,A)$ and define $l,r\colon E(A)\to A$ by $l(g)=g(1)$ and $r(g)=0$.  Note that there is $g\colon\bbZ_n\to A$ only if $n\cdot g(1)=0$, that is, only if $g(1)$ `has torsion'. The equations then force the quotient of $A$ to be torsion-free
(i.e., no element different than 0 has torsion). Since $E(\bbZ)=1$ and $E(\bbZ_2)$ is infinite, the epi $\bbZ\to\bbZ_2$ is not preserved. $\Alg{T,\E}$ is the category of torsion-free abelian groups.
\end{enumerate}
\end{example}

\medskip\noindent Let $\bb{A} = \tup{A,\alpha}$ be a $\T$-algebra.
We denote by $s_\bb{A}$ the coequalizer of $l^\sharp_A\circ\alpha$ and $r^\sharp_A\circ\alpha$
in $\Alg{\T}$ depicted in the following diagram (we suppress the forgetful functor
and denote by the same symbol both an algebra morphism and its underlying $\C$-arrow)
\begin{equation}\label{eq:def-s}
\xymatrix{
  \tup{TEA,\mu_{EA}} \ar@<2px>[r]^-{\eql_A^\sharp} \ar@<-2px>[r]_-{\eqr_A^\sharp} & \tup{TA,\alpha} \ar[r]^-{\alpha} & \tup{A,\alpha} \ar[r]^-{s_{\bb{A}}} & \tup{A/\E,\quotE{\alpha}}\,.
%  TE\bb{A} \ar@<2px>[r]^-{\eql^\sharp} \ar@<-2px>[r]_-{\right^\sharp} & TA \ar[r]^-{\alpha} & A \ar[r]^{-s_{\bb{A}} & A/\!E
}
\end{equation}
%For a free algebra $\tup{TX,\mu}$ we will write  $s_{TX}$ instead of $s_{\tup{TX,\mu}}$.
In the case $\C=\Set$, this entails that $\ker(s_\bb{A})$ is the \emph{congruence} generated by
the set $E_\bb{A} = \{(\alpha(l_A(e)),\alpha(r_A(e))\mid e\in EA\}$, i.e., by the least equivalence
relation on $A$ that includes $E_\bb{A}$ and is a subalgebra of $\tup{A,\alpha} \times \tup{A,\alpha}$.
In this sense, the kernel pair of a morphism always yields a congruence, and conversely, every
congruence relation on an algebra $\tup{A,\alpha}$ is the kernel of the corresponding quotient
homomorphism.
In general, the coequaliser~\eqref{eq:def-s} 
in $\Alg{\T}$ differs from the one obtained by applying the forgetful functor $U$ and then computing the coequaliser of 
$\alpha\circ l_A^\sharp$ and $\alpha\circ r_A^\sharp$ in $\Set$.
The coequalisers in $\Alg{T}$ and $\Set$ coincide if the equations
are reflexive in the sense that the two parallel arrows
%$\xymatrix@1{EA \ar@<2px>[r]^-{} \ar@<-2px>[r]_-{} & TA \ar[r]^{} & A}$
$\alpha\circ l_A$ and $\alpha\circ r_A$ from $EA$ to $A$
have a common section, and the forgetful functor $U$ preserves reflexive coequalisers.
If $T$ is finitary, then $U$ preserves reflexive coequalisers.
Moreover, we note that if $U$ preserves reflexive coequalisers
then $T$ preserves them too, but not every $\Set$-functor
preserves reflexive coequalisers, cf.~\cite[Example 4.3]{AKV00}.

A $\T$-algebra $\bb{A} = \tup{A,\alpha}$ is said to \emph{satisfy} $\E$ if the following
diagram commutes:
$$
\xymatrix{
  EA \ar@<2px>[r]^-{\eql_A} \ar@<-2px>[r]_-{\eqr_A} & TA \ar[r]^{\alpha} & A \, .
}
$$
By $\Alg{\T,\E}$ we denote the full subcategory of $\T$-algebras that satisfy $\E$.
As coequalisers are unique only up to isomorphism, we will choose $s$ such that for
all $\bb{A}\in \Alg{\T,\E}$, $s_{\bb{A}} = \id_\bb{A}$.
%We will use this to prove Lemma~\ref{lm:adj}.

\begin{lemma}\label{lm:adj}
The inclusion $V\colon \Alg{\T,\E}\to\Alg{\T}$ has a left-adjoint $H\colon \Alg{\T} \to \Alg{\T,\E}$
with unit $\ol{\eta}_\bb{A} = s_\bb{A}:\tup{A,\alpha} \to \tup{A/\E,\quotE{\alpha}}$ for all $\bb{A} \in \Alg{\T}$, and
counit $\ol{\epsilon}_\bb{A} = \id_\bb{A}$ the identity for all $\bb{A} \in \Alg{\T,\E}$.
In particular, $\Alg{\T,\E}$ is a full, reflective subcategory of $\Alg{\T}$.
\end{lemma}

\begin{proof}
We first show that for any $\bb{A} = \tup{A,\alpha}$ in $\Alg{\T}$,
$\tup{A/\E,\quotE{\alpha}}$ is an object in $\Alg{\T,\E}$.  Consider the following diagram:
\[
\xymatrix{
TEA \ar@(r,u)@<2px>[dr]^{\eql_{A}^\sharp} \ar@(r,u)@<-2px>[dr]_{\eqr_{A}^\sharp} & & \\
EA \ar[d]_-{Es_{\bb A}} \ar[u]^{\eta_{EA}}\ar@<2px>[r]^{\eql_{A}} \ar@<-2px>[r]_-{\eqr_{A}} & TA \ar[d]^{Ts_{\bb A}}\ar[r]^{\alpha} & A \ar[d]^{s_{\bb A}} \\
EA/\E \ar@<2px>[r]^-{\eql_{A/\E}} \ar@<-2px>[r]_-{\eqr_{A/\E}} & TA/\E \ar[r]^{\quotE{\alpha}} & A/\E
}
\]
The right-hand square commutes by the definition of $s_\bb{A}$, cf.~\eqref{eq:def-s}.
The left squares (for $l$ and $r$ respectively)
commute by naturality of $l$ and $r$. The upper two paths from $TEA$ to $A/\E$ commute by definition of $s_\bb{A}$.
From the above diagram we obtain
$\quotE{\alpha} \circ \eql_{A/\E} \circ E(s_\bb{A}) = \quotE{\alpha} \circ \eqr_{A/\E} \circ E(s_\bb{A})$ which implies
$\quotE{\alpha} \circ \eql_{A/\E} = \quotE{\alpha} \circ \eqr_{A}$ since $E(s_\bb{A})$ is epic (cf.~Assumption~\ref{assumps}).

It remains to show that for $\bb{B}=\tup{B,\beta}$ in $\Algeq{\T}{\E}$ and an algebra morphism $f:A\to B$ there is a unique algebra morphism $g:A/\E\to B$ such that $g\circ s_\bb A = f$. Since $\bb{B}$ satisfies the equations we know $\beta\circ l_B=\beta\circ r_B$, hence $f\circ\alpha\circ l_A=f\circ\alpha\circ r_A$. Since $s_\bb{A}:A\to A/\E$ is the coequalizer of $(\alpha\circ l_A,\alpha\circ r_A)$ the claim follows. The situation is illustrated here:
\[
\xymatrix{
EA \ar@<2px>[r]^-{~~\eql_{A}} \ar@<-2px>[r]_-{\eqr_{A}} \ar[d]_-{Ef} & TA \ar[d]_-{Tf}\ar[r]^{\alpha} & A \ar[d]_-{f} \ar[r]^-{s_\bb{A}} & A/\E \ar@{-->}[dl]^-{g}\\
EB \ar@<2px>[r]^-{\eql_{B}} \ar@<-2px>[r]_-{\eqr_{B}} & TB \ar[r]^{\beta} & B
}
\]
We have shown that $s_{\tup{A,\alpha}} \colon \tup{A,\alpha} \to \tup{A/\E,\quotE{\alpha}}$
is an $\Alg{\T,\E}$-reflection arrow for $\tup{A,\alpha}$.
By defining $H\colon \Alg{\T} \to \Alg{\T,\E}$ as
$H\tup{A,\alpha} = \tup{A/\E,\quotE{\alpha}}$, then $H$ is left adjoint to $V$,
and the unit of the adjunction is $\ol{\eta}=q$.
Now, since the unit and counit must satisfy
$V(\epsilon_{\bb{A}}) \circ s_{V\bb{A}} = id_{V\bb{A}}$, for all $\bb{A} \in \Alg{\T,\E}$,
it follows from $s_{V\bb{A}} = \id_\bb{A}$  and $VH=\Id_{\Alg{\T,\E}}$ that
$V(\epsilon_{\bb{A}}) = V(id_{\bb{A}})$, and hence
$\epsilon_{\bb{A}} = id_{\bb{A}}$.
%Since $V$ is a full inclusion, it follows from \cite[IV.3]{maclane}
%%that the counit $\ol{\epsilon}\colon HV \To \Id_{\Alg{\T,\E}}$ is a natural isomorphism.
%that we can choose the left adjoint $H$ so that the counit $\ol{\epsilon}\colon HV \To \Id_{\Alg{\T,\E}}$ becomes the identity
%on $\Alg{\T,\E}$.
\end{proof}

%\blue{XXX}
%It follows that $q_\bb{A}$ and, in particular, $q_{TX}$ are natural transformations from $T$ to $T$ \blue{FIX and make consistent with the stuff below; $q_T$ is not a good name}.
% Moreover,
By composition of adjoints, the functor $UV \colon \Alg{\T,\E}\to\Alg{\T}\to\C$ has a left adjoint given by $X \mapsto \tup{TX/\E,\quotE{(\mu_X)}}$. In what follows, we will write $T'X$ for $TX/\E$.
This allows the following definition.

\begin{definition}[Quotient monad]{\label{defi:quot-monad}}
Given a monad $\T = \tup{T,\eta,\mu}$ on $\C$ and $\T$-equations $\E$,
we define the \emph{quotient monad} $\T' = \tup{T',\eta',\mu'}$ as the monad on $\C$
arising from the composition of the adjunction $\tup{H,V,\ol{\eta}=s,\ol{\epsilon}=\id}$ of Lemma~\ref{lm:adj}
and the Eilenberg-Moore adjunction $\tup{G,U,\eta,\epsilon}$ of $\T$:\\[.6em]
\[\xymatrix{
\Alg{\T,\E} \ar@/^1.5pc/[r]^V \ar@{}[r]|{\top}
& \Alg{\T} \ar@/^1.5pc/[r]^U \ar@/^1.5pc/[l]^H \ar@{}[r]|{\top}
& \C \ar@/^1.5pc/[l]^G \ar@(ur,dr)^-{\T'}
}\]
\end{definition}

We define $q\colon T \To T'$ as the family of underlying $\C$-arrows of reflection arrows for free algebras, i.e.,
\begin{equation}\label{eq:def-q}
    q_X = U s_{\tup{TX,\mu_X}} \colon TX \to T'X
\end{equation}
Naturality of $q$ is clear, since $s$ is natural.
Next, we show that $q$ is a monad morphism from $\T$ to $\T'$. One way of doing so is to show that $q$ is a
coequaliser in the category of monads and monad morphisms. Kelly studied colimits in categories of monads,
and proved their existence in the context of some adjunction~\cite[Proposition 26.4]{kelly:verylong}; with a bit of effort one can instantiate
this to the adjunction constructed above.
For a self-contained presentation in this section, we do not invoke Kelly's results but instead prove directly the
part that shows the existence of a monad morphism. This is instantiated below to the adjunction of
the quotient monad.

\begin{lemma}\label{lm:mm-general}
  Let $\A$ be any subcategory of $\Alg{\T}$, and suppose the forgetful functor $U : \A \rightarrow \C$ has a left adjoint
  $F$, with unit and counit denoted by $\eta'$ and $\epsilon'$ respectively. Then
  \begin{enumerate}
    \item $F$ induces a natural transformation $\kappa : TUF \Rightarrow UF$ so that
    $\kappa \circ T\eta' : T \Rightarrow UF$ is a monad morphism.
    \item Precomposing the functor $\Alg{UF} \rightarrow \Alg{\T}$ induced by this monad morphism with the comparison functor
    $\A \rightarrow \Alg{UF}$ yields the inclusion $\A \rightarrow \Alg{\T}$.
  \end{enumerate}
\end{lemma}
\begin{proof}
  The functor $F$ sends any $\C$-object $X$ to a $T$-algebra structure on $UFX$; we define $\kappa_X$ to be that algebra structure.
  Naturality of $\kappa$ is immediate since $Ff$ is a $\T$-algebra homomorphism for any $\C$-arrow $f$.
  To see that $\kappa \circ T\eta'$ is a monad morphism, consider:
  $$
  \xymatrix{
    TT \ar[r]^{TT\eta'} \ar[d]^{\mu}
      & TTUF \ar[r]^{T\kappa} \ar[d]^{\mu_{UF}}
      & TUF \ar[r]^{T\eta'_{UF}} \ar[d]^\kappa
      & TUFUF \ar[d]^{\kappa_{UF}}\\
    T \ar[r]^{T\eta'}
      & TUF \ar[r]^{\kappa}
      & UF
      & UFUF \ar[l]_{U\epsilon'_F} \\
    \Id \ar[u]^{\eta} \ar[r]_{\eta'}
      & UF \ar[u]^{\eta_{UF}} \ar@{=}[ur]
      & &
  }
  $$
  The top left square commutes by naturality and the middle square since
  any component of $\kappa$ is an $\T$-algebra. For the right square we have
  $$
    \kappa = \kappa \circ TU\epsilon'_F \circ T\eta'_{UF}
    = U\epsilon'_F \circ \kappa_{UF} \circ T\eta'_{UF}
  $$
  where the first equality follows from the triangle identity $\id_{UF}= U\epsilon'_F \circ \eta'_{UF}$ (and functoriality), and the second from the fact that $\epsilon'_{FX}$ is a $\T$-algebra homomorphism from $\kappa_{UFX}$ to $\kappa_{X}$.
  The bottom left square commutes by naturality,
  and the triangle since $\kappa$ is an $\T$-algebra.

%  The top left square commutes by naturality and the middle square since
%  any component of $\kappa$ is an $\T$-algebra. For the right square we have
%   use that each component of $\epsilon'$ is an algebra morphism
%  and that \blue{$ T U\epsilon' F \circ T \eta'UF$ by a triangle identity. [Helle: Sory, I don't follow what is going on here....]} The bottom left square commutes by naturality,
%  and the triangle since $\kappa$ is an $\T$-algebra.

  For (2), we first note that the composite functor under consideration maps any $T$-algebra $\tup{A,\alpha}$ in $\A$
  to $U\epsilon'_{U\tup{A,\alpha}} \circ \kappa_A \circ T \eta'_A$. But we have
  $$
    \alpha = \alpha \circ TU\epsilon'_{\tup{A,\alpha}} \circ T\eta'_A
    = U\epsilon'_{U\tup{A,\alpha}} \circ \kappa_A \circ T\eta'_A
  $$
  where the first equality is a triangle identity and the second is the fact that $\epsilon'_{\tup{A,\alpha}}$ is an algebra morphism.
\end{proof}

\begin{remark}
Lemma~\ref{lm:mm-general} is a special case of what is known as the structure-semantics adjointness, which establishes an adjunction between (algebraic theories or) monads  over $\C$ (the structure) and `forgetful' functors $\A\to \C$ (the semantics), see \cite{lawvere:phd}, \cite{linton:functorial-semantics}, \cite{street:ftm} and, in particular, \cite[Theorem II.1.1]{dubuc:kan}.  
\end{remark}

\begin{corollary}
  $q \colon T \Rightarrow T'$ is a monad morphism.
\end{corollary}
\begin{proof}
  By Lemma~\ref{lm:mm-general}, we only need to show that $q$ coincides with $\kappa \circ T\eta'$, where $\eta'$ is
  the unit of the quotient monad. To this end consider the following diagram:
  $$
  \xymatrix{
    T \ar[r]^{T\eta'} \ar[dr]_{T\eta}
      & TUF \ar[r]^\kappa
      & UF \\
    &
    TT \ar[r]^{\mu} \ar[u]_{Tq}
      & T \ar[u]_q
  }
  $$
  Commutativity of the triangle follows from the definition of the quotient monad. For the square, notice that
  the components of $\kappa$ are simply the quotient algebras as constructed in the proof of Lemma~\ref{lm:adj},
  and $q$ is an algebra morphism by construction.
\end{proof}

\begin{remark}
As always, the monad morphism $q:T\To T'$ induces a functor
$$
  \Alg{\T'}\to\Alg{\T}.
$$
By Lemma~\ref{lm:mm-general} (2), the comparison $\Alg{\T,\E}\to\Alg{\T'}$ followed
by $\Alg{\T'}\to\Alg{\T}$ coincides with the inclusion $\Alg{\T,\E}\to\Alg{\T}$.
\end{remark}

%\blue{from here on its the same as before}

The above construction yields a monad $\T'$ given a set of operations and equations. Intuitively,
any monad which is isomorphic to $\T'$ is presented by these same operations and equations; this is captured
by the following definition.
\begin{definition}\label{def:monad-pres}
  Let $\Sigma$ be an endofunctor on $\C$, $\T_\Sig$ the free monad over $\Sig$, and $\T'$ the quotient monad of $\T_\Sig$ with respect to some  $\T_\Sig$-equations $\E$. A monad $\K=\tup{K,\theta,\nu}$
  is \emph{presented by $\Sigma$ and $\E$} if there is a monad isomorphism $i \colon T' \Rightarrow K$.
\defiEnd
\end{definition}
\begin{example}\label{ex:ism}
  The \emph{idempotent semiring monad} is defined by the functor mapping a set $X$ to
  the set $\Pow_\omega(X^*)$ of finite languages over $X$ and, for morphisms $f \colon X \to Y$ in $\Set$
  we define $\Pow_\omega(f^*)(L) = \bigcup \{ f(x_1) \cdots f(x_n) \mid x_1 \cdots x_n \in L \}$. Furthermore,
  $\eta_X \colon X \to \Pow_\omega(X^*)$ and
  $\mu_X \colon \Pow_\omega(\Pow_\omega(X^*)^*) \to \Pow_\omega(X^*)$ are given by
\[\begin{array}{rcl}
 \eta_X(x) &=& \{ x \}, \\[.3em]
 \mu_X(\mathcal{L}) &=& \bigcup_{L_1\cdots L_n \in \mathcal{L}} \{ w_1\cdots w_n \mid w_i \in L_i \}.
\end{array}
\]
  The idempotent semiring monad is presented by two constants $\cns{0}$ and $\cns{1}$, two binary operations $+$ and $\cdot$, and
  the idempotent semiring axioms. The witnessing isomorphism can easily be given based on the observation
  that every semiring term is equivalent with respect to the idempotent semiring equations to a sum of products of variables.
\defiEnd
\end{example}

Finally, we discuss conditions under which $\Alg{\T,\E}$ is isomorphic to $\Alg{\T'}$. In general this need not be the case due to the fact that despite both $\Alg{\T,\E}\to\Alg{\T}$ and $\Alg{\T}\to\C$ being monadic, their composition need not be monadic. This situation occurs in Example~\ref{exle:assumptions}(\ref{exle:assumptions:3}), where $\Alg{\T,\E}$ is torsion-free abelian groups but $\Alg{\T'}=\Alg{\T}$ is abelian groups \cite{Borceux:HB-vol2}.

A general remark in this situation is that, due to Beck's theorem, $\Alg{\T,\E}\to\C$ is monadic if $\Alg{\T,\E}\to\Alg{\T}$ is closed under regular epis, that is, if $A\in \Alg{\T,\E}$ implies $B\in \Alg{\T,\E}$ for regular epis $f:A\to B$ in $\Alg{\T}$. But we can do better in a situation of special interest.

We say that $\C$ is a finitary variety if $\C$ is the category of algebras for a finitary signature and equations, or, equivalently, if $\C$ is the category of algebras for a finitary monad on $\Set$. In particular, we have an adjoint situation $F^\C\dashv U^C:\C\to\Set$. A signature is a functor $\bb{N}\to\Set$, assigning to each `arity' $n$ a set of operation symbols. An endofunctor on $\C$ is said to be generated by a signature $S:\bb{N}\to\Set$ if it is the left-Kan extension of $F^\C S$ along $F^C$, that is, if it is of the form $A\mapsto \coprod_{n\in\bb N}  \C(F^\C n,A)\bullet F^\C Sn$. (Note that we cannot use such endofunctors to imitate Example~\ref{exle:assumptions}(\ref{exle:assumptions:3}) which relies on  the $\bb Z_n$ not being free algebras of the form $F^\C n$.)

\begin{proposition}
If in the situation described in Definition~\ref{def:monad-pres} the category $\C$ is a finitary variety, and $\Sigma$ and $E$ are endofunctors on $\C$ generated by signatures, then $\Alg{\T,\E}$ is isomorphic to $\Alg{\T'}$.
\end{proposition}

\begin{proof}
This is an instance of Theorem~4.4 of \cite{VelebilK11}.
\end{proof}

%------------------------------------------------------------------

\section{Distributive Laws and Bialgebras}
\label{sec:dls}

We briefly recall the basic definitions of distributive laws and bialgebras;
for a more thorough introduction we refer to \cite{Klin11,Bartels:PhD,TuriPlotkin:LICS-GSOS}.

\subsection{Basic Definitions}
\label{ssec:dls-basic}

Let $\T = \tup{T,\eta,\mu}$ be a monad on a category $\C$,
and $F$ an endofunctor on $\C$.
A \emph{distributive law} $\lambda$ of the monad $\T$ over the functor $F$
is a natural transformation $\lambda \colon TF\To FT$
which is compatible with the monad structure,
meaning that $\lambda\circ\eta_F = F\eta$ and
$\lambda\circ \mu_F = F\mu\circ\lambda_T\circ T\lambda$,
i.e., for all $X$ the following diagrams commute:
\[
\xymatrix@R=1.5em@C=1.5em{
FX \ar[rr]^-{\eta_{FX}} \ar[ddrr]_-{F\eta_{X}}
&& TFX \ar[dd]^-{\lambda_X}
\\
&  \ar@{}[r]^(.4){(\text{unit.})\lambda} & &&&
\\
&& FTX &&
}
\xymatrix@C=4em@R=3em{
T^2FX \ar[d]_-{\mu_{FX}} \ar[r]^-{T\lambda_X} \ar@{}[drr]|{(\text{mult.})\lambda}
& TFTX \ar[r]^-{\lambda_{TX}}
& FT^2 X \ar[d]^-{F\mu_X}
\\
TFX \ar[rr]^-{\lambda_X} && FTX
}
\]
We recall that every distributive law $\lambda\colon TF \To FT$
corresponds to a \emph{lifting} $F_\lambda$ of $F$ to the category of
$\T$-algebras (see, e.g.,~\cite{Johnstone:Adj-lif,Klin11}), defined as
\begin{equation}\label{eq:F-lambda}
  F_\lambda \tup{A, \alpha} = \tup{FA, F \alpha \circ \lambda_A} \qquad \qquad F_\lambda(f) = Ff
\end{equation}
Note that the compatibility of $\lambda_X$ with $\mu_X$ means precisely that $\lambda_X$ is a $\T$-algebra homomorphism from $\tup{TFX, \mu_{FX}}$ to $F_\lambda\tup{TX, \mu_X}$.

% Coalgebras
An \emph{$F$-coalgebra} is a pair $\tup{X,c}$ where $X$ is a $\C$-object and
$c \colon X \to FX$ is a $\C$-arrow. An \emph{$F$-coalgebra morphism}
from $\tup{X,c}$ to $\tup{Y,d}$ is an arrow $f\colon X \to Y$
such that $d\circ f = Tf \circ c$.
%
% Bialgebras
Given a distributive law $\lambda$ of $\T$ over $F$,
a \emph{$\lambda$-bialgebra} $\tup{X,\alpha,\beta}$ consists of a carrier $X$,
a $\T$-algebra $\alpha\colon TX\to X$ and an $F$-coalgebra $\beta\colon X \to FX$ such that
$\beta\circ\alpha = F\alpha\circ\lambda_X\circ T\beta$.
A~\emph{morphism of $\lambda$-bialgebras} from $\tup{X_1,\alpha_1,\beta_1}$
to $\tup{X_2,\alpha_2,\beta_2}$ is an arrow $f\colon X_1 \to X_2$
which is both a $\T$-algebra homomorphism and an $F$-coalgebra morphism.

The following results are well known (see, e.g.,~\cite{TuriPlotkin:LICS-GSOS,Bartels:PhD,Klin11}).
If $\tup{Z,\zeta}$ is a final $F$-coalgebra, then
a distributive law $\lambda\colon TF \To FT$ yields a
final $\lambda$-bialgebra $\tup{Z,\alpha,\zeta}$
where  $\alpha\colon TZ \to Z$ is defined by coinduction
from the $F$-coalgebra $\tup{TZ,\lambda_Z\circ T\zeta}$.

We will need the notion of distributive laws of monads over \emph{copointed} functors.
A~copointed functor is a pair $\tup{F, \epsilon}$ where $F$ is an endofunctor and $\epsilon \colon F \Rightarrow \Id$
a natural transformation. A distributive law of $\T$ over $\tup{F, \epsilon}$
is a distributive law of $\T$ over $F$ additionally satisfying $\epsilon_T \circ \lambda = T\epsilon$.
For any endofunctor $F$ on a category $\C$ with products, the \emph{cofree copointed functor generated by $F$}
is the pair $\tup{\Id\times F,\pi_1\colon \Id\times F \to \Id}$
where $\pi_1$ is the natural left-projection.

When $\T = \T_\Sig$ is the free monad generated by a (signature) functor $\Sig$,
then distributive laws involving $\T$ can be reduced to ``plain''
natural transformations using recursion, namely,
there is a 1-1 correspondence between distributive laws
$\lambda\colon T_\Sig F \To FT_\Sig$ of $\T_\Sig$ over $F$ and natural transformations
$\rho\colon \Sig F \To FT_\Sig$ (cf.~\cite{TuriPlotkin:LICS-GSOS,Bartels:PhD}).
Such a $\rho$ corresponds to a specification format of operational rules,
and is sometimes referred to as a \emph{simple SOS specification}.
Similarly, for cofree copointed functors,
if $\T_\Sig$ is freely generated by $\Sig$,
then there is a 1-1 correspondence between distributive laws
$\lambda\colon T_\Sig(\Id\times F) \To (\Id\times F)T_\Sig$ of
$\T_\Sig$ over $\tup{\Id\times F,\pi_1}$ and natural transformations
$\rho\colon \Sig(\Id\times F) \To FT_\Sig$ (cf.~\cite{LPW2004:cat-sos,Jacobs:bialg-dfa-regex}).
Such a natural transformation $\rho$ is also referred to as an
\emph{abstract GSOS specification} since it generalises the GSOS-format
for labelled transition systems where $F=(\Pow_\omega(-))^A$,
cf.~\cite{Bartels:PhD,TuriPlotkin:LICS-GSOS}.
In what follows, we will generally omit $\Sigma$-subscripts on free monads in order to keep notation uncluttered.

%------------------------------------------------------------------

\subsection{Solutions to Corecursive Equations}
\label{ssec:corecursive-eqs}
%- corecursive equations
%  special cases: linear systems, context-free systems
%- solutions via $\lambda$-coinduction

An important application of distributive laws is in
solving \emph{corecursive equations} which are arrows of the type
$\phi\colon X \to FTX$ where $F$ is a functor and $T$ is (the functor component of) a monad.
These include many interesting and useful structures such as
linear and context-free systems of behavioural differential equations~\cite{Rut03:TCS-bde,WBR:CF-CALCO}, as well as
linear, nondeterministic and weighted automata
cf.~\cite{Jacobs:bialg-dfa-regex,SilvaBBR10}.
These are all instances of \emph{$\T$-automata}~\cite{Jacobs:bialg-dfa-regex}
(where $\T$ is a monad on $\Set$)
which have the type
$X \to B \times (TX)^A$ where
$A$ is a set and $B$ carries a $\T$-algebra $\beta\colon TB \to B$,
i.e., in particular, $F = B \times (-)^A$ whose final coalgebra
carrier is $B^{A^*}$.

In the presence of a distributive law $\lambda\colon TF \To FT$
one obtains a \emph{$\lambda$-coinduction principle} \cite{Bartels:PhD}
which provides unique solutions in the final $\lambda$-bialgebra
$\tup{Z,\alpha,\zeta}$
to corecursive equations of the form
$\phi\colon X \to FTX$.
Ordinary coinduction is the special case where $\T$ is the identity monad.
Formally, a solution to $\phi\colon X \to FTX$
in a $\lambda$-bialgebra $\tup{A,\alpha,\beta}$
is an arrow $f\colon X \to A$ such that
\begin{equation}\label{eq:solution}
\vcenter{
\xymatrix{
X\ar[d]_-{\phi} \ar[rr]^{f} && A\ar[d]^{\beta} \\
FTX \ar[r]^-{FTf} & FTA \ar[r]^{F\alpha} & FA
}
}
\end{equation}
commutes.
More precisely, $\lambda$-coinduction is coinduction in the category of
$\lambda$-bialgebras, and we have the following fact.

\begin{proposition}[Lemmas~4.3.3~and~4.3.4 of \cite{Bartels:PhD}]
\label{prop:solution}
Let $\phi\colon X \to FTX$ be a corecursive equation.
Taking
$\ext{\phi}{\lambda} = F\mu_X \circ \lambda_{TX} \circ T\phi$
then $\tup{TX,\mu_X,\ext{\phi}{\lambda}}$ is a $\lambda$-bialgebra,
and $\eta_X \colon X \to TX$ is a solution of $\phi$.
Moreover,
for any $\lambda$-bialgebra $\tup{A,\alpha,\beta}$,
there is a 1-1 correspondence
between solutions of $\phi$ in $\tup{A,\alpha,\beta}$
and $\lambda$-bialgebra morphisms from
$\tup{TX,\mu_X,\ext{\phi}{\lambda}}$ to $\tup{A,\alpha,\beta}$.
\end{proposition}

\begin{comment}
So-called $\T$-automata \cite{Jacobs:bialg-dfa-regex,SilvaBBR10}
are corecursive equations where
$F = B \times (-)^A$ for a set $A$ and a
$\T$-algebra $\beta\colon TB \to B$.
These include linear (stream) automata \cite{BBBRS:lwa,Rut03:TCS-bde}
(where $B$ is a field and $\T$ is the vector space monad);
context-free grammars \cite{WBR:coalg-CFL}
(where $B$ is the Boolean semiring and $\T$ is monad of idempotent semirings);
and weighted automata
(where $B$ is a semiring and $\T$ is the semimodule monad).
For $\T$-automata,
$\lambda$-coinduction yields ($B$-weighted) language
semantics for linear, nondeterministic and weighted automata,
namely, solutions are maps into the final $F$-coalgebra
which has carrier $B^{A^*}$.
\end{comment}

A ``pointwise distributive law'' $\lambda$ for $\T$-automata
can be obtained (cf.~\cite{Jacobs:bialg-dfa-regex,Jacobs06}) by taking
$\lambda_X = (\beta\times\mathsf{st})\circ \tup{T\pi_1,T\pi_2}$
where $\mathsf{st}\colon T\circ (-)^A \To (-)^A\circ T$
is the strength natural transformation.
This $\lambda$ is called ``pointwise'', since the algebra structure induced on the carrier $B^{A^*}$ of the final $B \times (-)^A$-coalgebra is the pointwise extension of $\beta\colon TB \to B$.
In the context-free and streams examples below, however,
the desired algebraic structure on $B^{A^*}$ uses the convolution product
which is not the pointwise extension of the semiring product of $B$.
So for these examples a different $\lambda$ must be given.

%------------------------------------------------------------------

\section{Quotients of Distributive Laws}
\label{sec:quotients-of-dls}

In Section~\ref{sec:quot-monad} we saw how equations give rise to
quotients of algebras, and we gave a construction of the resulting
quotient monad.
In this section, we investigate conditions under which
distributive laws and equations give rise to quotients of distributive laws.

As before, let $\Sig$ be a functor
generating the free monad $\T = \tup{T,\eta,\mu}$,
and let $\E = \tup{E,\eql,\eqr}$ be $\T$-equations with the associated
quotient monad $\T' = \tup{T',\eta',\mu'}$.

\subsection{Distributive Laws over Plain Behaviour Functors}
\label{ssec:dls-plain-F}

In this subsection, we assume that
$\lambda \colon TF \To FT$ is a distributive law
of a monad $\T$ over a plain behaviour functor $F$.
We will provide a condition on $\lambda$ and the
equations $\E$ that ensures that we get a distributive law
$\lambda' \colon T'F \To FT'$ for the quotient monad.
To this end, it is convenient to use the notion of a morphism
of distributive laws from~\cite{PW02:monad-comonad,Watanabe:cmcs2002}.

\begin{definition}\label{defi:morphism-of-dl}
Let $\tup{T,\eta,\mu}$ and $\tup{K,\theta,\nu}$ be monads,
and let $\lambda\colon TF \To FT$ and $\kappa\colon KF \To FK$
be distributive laws.
A natural transformation $\tau\colon T \To K$ is a \emph{morphism of distributive laws} from $\lambda$ to $\kappa$
(notation $\tau\colon \lambda \To \kappa$)
if $\tau$ is a monad morphism and the following square commutes:
\begin{equation}\label{eq:morphism-dl-square}
\xymatrix{
TF \ar@{=>}[d]_-{\lambda}  \ar@{=>}[r]^-{\tau{F}}
& KF \ar@{=>}[d]^-{\kappa}
\\
FT \ar@{=>}[r]^-{F\tau}
& FK
}
\end{equation}
\defiEnd
\end{definition}

We note that there are generalisations of the above definition
that allow natural transformations between behaviour functors,
cf.~\cite{Watanabe:cmcs2002}. For our purposes, we do not need
to change the behaviour type.

\begin{definition}\label{defi:lambda'}
We say that \emph{$\lambda \colon TF \To FT$  preserves (equations in) $\E$} if
for all $X$ in $\C$:
\begin{equation}\label{eq:preservation-E}
\xymatrix{
  EFX \ar@<2px>[r]^-{\eql_{FX}} \ar@<-2px>[r]_-{\eqr_{FX}} & TFX \ar[r]^{\lambda_X} & FTX \ar[r]^{Fq_{X}} & FT'X
}
\end{equation}
commutes.
%We say that \emph{$\lambda \colon TF \To FT$  preserves (equations in) $E$}
%if for all $g\colon V \to FX$, and for all $s,t \in TV$:
%\begin{equation}\label{eq:preserve-eq}
%s \;E\; t \quad\Ra\quad Fq_{X}(\lambda_X(Tg(s))) \;=\; Fq_{X}(\lambda_X(Tg(t))).
%\end{equation}
\defiEnd
\end{definition}

In $\Set$, %assuming that $T$ is finitary, 
preservation of equations can be conveniently formulated in terms of relation lifting.
The $F$-lifting of a relation $R \sse Y \times Y$ is defined as
\[\Rel{F}(R) = \{ \tup{F\pi_1(u),F\pi_2(u)} \in FY \times FY \mid u \in F(R)\} \, .
\]
For any set $X$, we denote by $\Eq{X}$ the congruence $\ker(q_X)$ on $TX$ generated by the equations.
If the lifting $\Rel{F}$ preserves inverse images, then it preserves kernel relations.\footnote{The proof of this for polynomial functors in \cite[Lemma 3.2.5(i)]{BJ:Coalg-book-draft} goes through for arbitrary $F$ under the assumption of preservation of inverse images, since for any $F$, the lifting $\Rel{F}$ preserves diagonals.
In particular, $\Rel{F}$ preserves inverse images if $F$ preserves weak pullbacks (see e.g.,~\cite[Proposition 4.4.3]{BJ:Coalg-book-draft}).}
This means that $\Rel{F}(\Eq{X}) = \ker(Fq_X)$, and 
hence equation~\eqref{eq:preservation-E} is
satisfied if for every set $X$ and every $b \in EFX$:
\begin{equation}\label{eq:preserve-eq-lift}
\lambda_X \circ l_{FX}(b) \;\Rel{F}(\Eq{X})\; \lambda_X \circ r_{FX}(b).
\end{equation}

%Equation \eqref{eq:preservation-E}
%can be conveniently formulated in terms of relation
%lifting as
%\begin{equation}\label{eq:preserve-eq-lift}
%s \;E\; t \quad\Ra\quad \lambda_X(Tg(s)) \;\Rel{F}(\Eq{X})\; \lambda_X(Tg(t)).
%\end{equation}
%where the $F$-lifting of a relation $R \sse Y \times Y$ is defined as
%\[\Rel{F}(R) = \{ \tup{F\pi_1(u),F\pi_2(u)} \in FY \times FY \mid u \in F(R)\}
%\]
%and noticing that $u \; \Rel{F}(\Eq{X}) \; v$ iff $Fq_{X}(u) = Fq_{X}(v)$.

We now come to our main result. 
On the one hand, item (2) of Theorem~\ref{thm:lambda'-dlaw} below gives us a distributive law for the quotient monad and the useful consequences that follow from it such as, e.g., the solution of recursive equations as discussed in Sections \ref{ssec:corecursive-eqs} and \ref{sec:quotients-of-bialgebras}. On the other hand, condition (1) in the form of \eqref{eq:preservation-E} is amenable to explicit calculations as shown in Examples~\ref{ex:stream-calc}, \ref{exle:not-preserve}, and \ref{ex:CFG-lambda}.

\newcommand{\thmlambdaprime}{
The following are equivalent.
\begin{enumerate}
\item $\lambda\colon TF \To FT$ preserves equations in $\E$.
\item There is a (unique) distributive law $\lambda'\colon T'F \To FT'$
such that $q \colon T \To T'$
is a morphism
of distributive laws from $\lambda$ to $\lambda'$.
\end{enumerate}
}
\begin{theorem}\label{thm:lambda'-dlaw}
\thmlambdaprime
\end{theorem}

%In order to prove Theorem~\ref{thm:lambda'-dlaw}, we first prove a lemma.

%\begin{lemma}\label{lm:pres-eq-gsharp}
%If $\lambda \colon TF \To FT$ preserves equations in $E$ then
%for all valuations $g\colon V \to TFX$, and for all $s,t \in TV$:
%\begin{equation}\label{eq:pres-eq-gsharp}
%s \;E\; t \quad\Ra\quad \lambda_X(g^\sharp(s)) \;\Rel{F}(\Eq{X})\; \lambda_X(g^\sharp(t)).
%\end{equation}
%\end{lemma}
%\begin{proof}
%Let $g\colon V \to TFX$ be any valuation. Define $v \colon V \rightarrow FTX$ as $v = \lambda_X \circ g$.
%$$
%\xymatrix{
%  E \ar@<2px>[r]^{\pi_1} \ar@<-2px>[r]_{\pi_2}
%    & TV \ar[r]^{Tg} \ar[dr]^{Tv} \ar@(ur,ul)[rr]^{g^\sharp}
%      & TTFX \ar[d]^{T\lambda_X} \ar[r]^{\mu_{FX}} & TFX \ar[r]^{\lambda_X} & FTX \ar[r]^{Fq_{X}} & FT'X\\
%    & & TFTX \ar[r]^{\lambda_{TX}} & FTTX \ar[ur]^{F\mu_X} \ar[r]^{Fq_{X}} & FT'TX \ar[r]^{FT'q_{X}} & FT'T'X \ar[u]^{F\mu'_X}
%}
%$$
%The right half of the rectangle commutes by the fact that $q$ is a monad map and functoriality,
%the left half commutes by compatibility of $\lambda$ with $\mu$. The lower two paths from $E$ to $FT'TX$ (and thus on to $FT'X$)
%commute by assumption that $E$ preserves equations. The triangle commutes by definition of $v$ and functoriality, and the
%upper bubble commutes by definition of $g^\sharp$.
%Thus we obtain that $Fq_{X} \circ \lambda_X \circ g^\sharp \circ \pi_1 = Fq_{X} \circ \lambda_X \circ g^\sharp \circ \pi_2$,
%which proves~\eqref{eq:pres-eq-gsharp}.
%
%\end{proof}

% Proof of \ref{thm:lambda'-dlaw}

%\begin{proofof}{Theorem~\ref{thm:lambda'-dlaw}}
\begin{proof}
We first show that \eqref{eq:preservation-E} extends to the following:
\begin{equation}\label{eq:preservation-cong}
  \xymatrix{
    TEFX \ar@<2px>[r]^{\eql_{FX}^\sharp} \ar@<-2px>[r]_{\eqr_{FX}^\sharp} &  TFX \ar[r]^{\lambda_X} & FTX \ar[r]^{Fq_{X}} & FT'X
  }
\end{equation}
To obtain \eqref{eq:preservation-cong}
it suffices to show that $Fq_{X}\circ\lambda_X$ is a $\T$-algebra homomorphism,
since then $Fq_{X}\circ\lambda_X \circ \eql_{FX}^\sharp$ and $Fq_{X}\circ\lambda_X \circ \eqr_{FX}^\sharp$
are $\T$-algebra homorphisms extending $Fq_{X}\circ\lambda_X \circ \eql_{FX}$ and $Fq_{X}\circ\lambda_X \circ \eqr_{FX}$, respectively.
Since these latter two are equal
due to \eqref{eq:preservation-E}, and homomorphic extensions are unique,
we then get \eqref{eq:preservation-cong}.

We now show that $Fq_{X} \circ \lambda_X$ is a $\T$-algebra homomorphism.
Let $F_\lambda$ be the lifting of $F$ to the category of $\T$-algebras,
and recall that $\lambda_X$ is a $\T$-algebra homomorphism from $\tup{TFX, \mu_{FX}}$ to $F_\lambda\tup{TX, \mu_X}$ (cf.~Section~\ref{ssec:dls-basic}).
Since also  $\tup{TX, \mu_X} \xrightarrow{q_{X}} \tup{T'X, \mu'_X \circ q_{T'X}}$ is a $\T$-algebra homomorphism,
 by applying the lifting $F_\lambda$ we obtain a $\T$-algebra homomorphism
$$
\xymatrix{F_\lambda\tup{TX, \mu_X} \ar[r]^-{Fq_{X}} & F_\lambda \tup{T'X, \mu'_X \circ q_{T'X}}} \, .
$$
Thus $Fq_{X} \circ \lambda_X$ is a $\T$-algebra homomorphism from the free $\T$-algebra $\tup{TFX,\mu_{FX}}$.

\begin{comment}
Consider the following diagram:
\begin{equation}
  \xymatrix{
    TTFX \ar[r]^{T\lambda_X} \ar[ddd]^{\mu_{FX}}  & TFTX \ar[r]^{TFq_{X}} \ar[d]^{\lambda_{TX}}   & TFT'X \ar[d]^{\lambda_{T'X}} \\
                                                  & FTTX \ar[r]^{FTq_{X}} \ar[dd]^{F\mu_X}        & FTT'X \ar[d]^{Fq_{T'X}} \\
                                                  &                                             & FT'T'X \ar[d]^{F \mu'_X} \\
    TFX \ar[r]^{\lambda_X}                        & FTX \ar[r]^{Fq_{X}}                           & FT'X
  }
\end{equation}
where $\mu'$ is the multiplication of the quotient monad (Definition~\ref{defi:quot-monad}).
Commutativity of the left rectangle follows by the compatibility of $\lambda$ with $\mu$, the upper right square
commutes by naturality of $\lambda$ and the lower right rectangle commutes by the fact that $q$ is a monad map, and functoriality.
So $Fq_{X} \circ \lambda_X$ is a $\T$-algebra homomorphism.
\end{comment}

This proves that~\eqref{eq:preservation-cong} commutes.
Now, by the universal property of the coequalizer $q_{FX}$ there is a (unique) algebra homomorphism $\lambda'_X \colon T'FX \rightarrow FT'X$
such that $\lambda'_X \circ q_{FX} = Fq_{X} \circ \lambda_X$:
\begin{equation}\label{eq:lambda'-square}
\xymatrix{
TE{FX} \ar@<2px>[r]^{\eql_{FX}^\sharp} \ar@<-2px>[r]_{\eqr_{FX}^\sharp}
 & TFX \ar[d]_-{\lambda_X}  \ar[r]^-{q_{FX}}
 & T'FX \ar@{-->}[d]^-{\lambda'_X}
\\
& FTX \ar[r]^-{Fq_{X}}
& FT'X
}
\end{equation}
%\marginpar{We should argue that $\lambda'$ is natural, although it seems evident.}
The naturality of $\lambda'$ follows from \eqref{eq:lambda'-square}, and
the naturality of $\lambda$ and $q$.
Due to the commutativity of the square in \eqref{eq:lambda'-square}, $q$ is a morphism
of distributive laws from $\lambda$ to $\lambda'$
once we show that $\lambda'$ is, in fact, a distributive law.

The unit law for $\lambda'$ holds
due to the unit law for $\lambda$ and
\eqref{eq:lambda'-square}:
\begin{equation}\label{eq:lambda'-unit}
\xymatrix{
FX \ar[r]^-{\eta_{FX}} \ar[dr]_-{F\eta_{X}}
& TFX \ar[r]^-{q_{FX}} \ar[d]^-{\lambda_X}
& T'FX \ar[d]^-{\lambda'_X} \ar@{}[dl]|{\eqref{eq:lambda'-square}}
\\
& FTX \ar[r]_-{Fq_{X}} & FT'X
}
\end{equation}
Multiplication law for $\lambda'$:
\begin{equation}\label{eq:lambda'-mult}
\xymatrix@R=1.4cm@C=2.3cm{
TFX \ar[rr]^-{\lambda_X}
    \ar@/_5pc/[dddd]_-{q_{FX}}
    \ar@{}[drr]|{(\text{mult.})\lambda}
&& FTX  \ar@/^5pc/[dddd]^-{Fq_{X}}
\\
T^2FX \ar[u]_-{\mu_{FX}} \ar[d]^-{q_{TFX}}
      \ar[r]^-{T\lambda_X}
%      \ar@{->>}[d]^-{q_{FX}}
      \ar@{}[dr]|{(\text{nat.})q}
& TFTX\ar[r]^-{\lambda_{TX}} \ar@{->>}[d]^-{q_{FTX}} \ar@{}[dr]|{\eqref{eq:lambda'-square}_{TX}}
& FT^2X \ar[u]_-{F\mu_X} \ar@{->>}[d]^-{Fq_{TX}}
\\
T'TFX \ar[d]^-{T'q_{FX}}
      \ar[r]^-{T'\lambda_X}
%      \ar@{->>}[d]^-{T'q_{FX}}
      \ar@{}[dr]|{T'\eqref{eq:lambda'-square}}
& T'FTX \ar[r]^-{\lambda'_{TX}}
      \ar@{->>}[d]^-{T'Fq_{X}}
      \ar@{}[dr]|{\text{(nat.)}\lambda'}
& FT'TX  \ar@{->>}[d]^-{FT'q_{X}}
\\
T'T'FX
      \ar[r]^-{T'\lambda'_X}
      \ar[d]_-{\mu'_{FX}}
& T'FT'X \ar[r]^-{\lambda'_{T'X}} %\ar@/^.8pc/[u]^-{T'Fr_X}
& FT'T'X  \ar[d]^-{F\mu'_X} %\ar@/^.8pc/[u]^-{FT'r_X}
\\
T'FX \ar[rr]^-{\lambda'_X}
&& FT'X
}
\end{equation}
The small upper-left square commutes by naturality of $q$.
%\mycomment{Need to show $q$ natural!}
The small lower-left square
commutes by applying $T'$ to \eqref{eq:lambda'-square}.
The outer crescents commute since $q$ is a monad morphism, and the outermost
part does due to \eqref{eq:lambda'-square}. Finally, use that by naturality of $q$,
$T'q_{FX}\circ q_{TFX} = q_{T'FX}\circ Tq_{FX}$, which by Assumption~\ref{assumps}
is an epi, and hence can be right-cancelled to yield commutativity of the lower rectangle as desired.
%\mycomment{Maybe we can get away with the fact that $Tq$ is epi in $\Alg{T}$, which always holds}
%Algebraically:
%\[\begin{array}{rcl}
%&& F\mu'_X \circ \lambda'_{T'X} \circ T'\lambda'_X
%\\
%& = &
%F\mu'_X \circ (
%    FT'q_{X}\circ\lambda'_{TX}\circ T'Fr_X) \circ
%    (T'Fq_{X}\circ T'\lambda_X \circ T'r_{FX})
%\\
%& = &
% \\
%\end{array}\]

The implication from 2 to 1 follows from
the fact that \eqref{eq:lambda'-square}
implies  \eqref{eq:preservation-E}.
\end{proof}

\newcommand{\Ran}{\mathsf{Ran}}
\begin{remark}
Street~\cite{street:ftm} investigates the 2-category where monads are objects and distributive laws $\lambda\colon SF\To FT$, called monad functors, are the 1-cells $(\lambda,F)\colon T\to S$. Given a monad $T$, the right Kan-extension $\Ran_F FT$ of $FT$ along $F$ is again a monad. Morever, distributive laws $SF\To FT$ are in 1-1 correspondence to monad morphisms $S\To \Ran_F FT$, cf.~\cite[Theorem 5]{street:ftm}. 

In this setting, if the free monad $T_E$ over $E$ exists, then the equations can be expressed more abstractly as a parallel pair of monad morphisms
$\xymatrix{T_E
\ar[]!<0ex,0.35ex>;[r]!<0ex,0.35ex>
\ar[]!<0ex,-0.35ex>;[r]!<0ex,-0.35ex>&T}$, 
and the distributive law $\lambda\colon TF\To FT$ satisfies the equations iff 
its transpose $T\To\Ran_F FT$ does, that is, iff 
$\xymatrix{T_E
\ar[]!<0ex,0.35ex>;[r]!<0ex,0.35ex>
\ar[]!<0ex,-0.35ex>;[r]!<0ex,-0.35ex>&T
\ar[r]
&
\Ran_F FT
\ar[r]
&
\Ran_F FT'
}
$ 
commutes.
Since $T'$ is a coequaliser of monads, this induces a monad morphism $T'\To\Ran_F FT'$, which, after transposing, gives a distributive law
$\lambda': T'F\To FT'$. We thank Neil Ghani for this conceptually elegant argument of the equivalence stated in Theorem~\ref{thm:lambda'-dlaw}.
We note that the above elementary proof does not require the existence of the free monad over $E$, and moreover, it avoids introducing more abstract definitions.
\end{remark}

\begin{remark}
Using that distributive laws
correspond to functor liftings
on $\T$-algebras (cf.~\eqref{eq:F-lambda}), % (see e.g.~\cite{Johnstone:Adj-lif})
the distributive law $\lambda'$ in
Theorem~\ref{thm:lambda'-dlaw} exists if and only if the functor
$F_\lambda$ restricts to $\T'$-algebras.
A similar statement for the case when $F$ is a monad is made in
\cite[Corollary~3.4.2]{Manes:MonadComp}.
\end{remark}

As a corollary we obtain the analogue of Theorem~\ref{thm:lambda'-dlaw}
for monads presented by operations and equations.

\newcommand{\corpresdl}{
  Suppose $\K = \tup{K,\theta,\nu}$ is presented by operations $\Sigma$ and equations $\E$
  with natural isomorphism $i\colon T' \To K$,
  and suppose
  we have a distributive law $\lambda \colon TF \Rightarrow FT$ of $\T$ over $F$.
  Then there exists a unique distributive law $\kappa \colon KF \To FK$
  of $\K$ over $F$
  such that $i \circ q \colon \lambda \To \kappa$ is a morphism of distributive laws.
}
\begin{corollary}\label{cor:pres-dl}
  \corpresdl
\end{corollary}
\begin{proof}
The distributive law $\kappa \colon KF \To KF$ is defined as $\kappa = Fi \circ \lambda \circ i^{-1}$. The proof
proceeds by checking that $\kappa$ indeed satisfies the defining axioms of a distributive law, which is an easy but
tedious exercise.
\end{proof}

Theorem~\ref{thm:lambda'-dlaw}
says that if $\lambda$ preserves the equations $\E$,
then we can \emph{present $\lambda'$ as ``$\lambda$  modulo equations''}.
We illustrate this with an example.

\begin{example}[Stream calculus]\label{ex:stream-calc}
Behavioural differential equations are used extensively in \cite{Rut03:TCS-bde,Rut05:MSCS-stream-calc}
to define streams and stream operations. Here, the behaviour functor is $FX = \bbR \times X$
whose final coalgebra $\tup{\bbR^\omega,\zeta}$ consists of streams over the real numbers
together with the map $\zeta(\sigma) = \tup{\sig(0),\sig'}$
which maps a stream $\sig$ to its initial value $\sig(0)$ and derivative $\sig'$.

Consider the following system of behavioural differential equations
where $[a]$, $\X$, $\sigma$ and $\tau$ denote streams over the real numbers.
\begin{equation}\label{eq:sde-example}
\begin{array}{rclcrcl}
[a](0) &=& a, & \quad & [a]' &=& [0], \qquad\qquad\qquad\qquad \forall a \in \bbR\\
\X(0) &=& 0, && \X' &=& [1],\\
(\sig+\tau)(0) &=& \sig(0)+\tau(0), && (\sig+\tau)' &=& \sig'+\tau',\\
(\sig\times\tau)(0) &=& \sig(0)\cdot\tau(0), & &
(\sig\times\tau)' &=& (\sig'\times [\tau(0)]) + ((\sig'\times (\X \times \tau')) +\\
&&&&&& \hfill ([\sig(0)]\times\tau'))
%(\sig\times\tau)' &=& (\sig'\times\tau) + ([\sig(0)]\times\tau')
%\\
\end{array}
\end{equation}
The behavioural differential equations in \eqref{eq:sde-example} define the constant streams
$[a] = (a,0,0,\ldots)$ for all $a \in \bbR$,
$\X=(0,1,0,0,\ldots)$,
pointwise addition and convolution product of streams.
Note that the convolution product is here defined differently than in \cite{Rut03:TCS-bde,Rut05:MSCS-stream-calc}.
We explain this choice at the end of the example.

Since we are defining $\bbR$ many streams $[a]$, one constant stream $\X$ and two binary operations ($+$ and $\times)$, the signature functor is
$\Sig(X) = \bbR + 1 + (X \times X) + (X \times X)$, and \eqref{eq:sde-example} corresponds to a natural transformation $\rho\colon \Sig F \To F T$ where $T$ is the functor part of the free monad $\T$ over $\Sig$ (that is, $TX$ is the set of all $\Sig$-terms over variables in $X$). The components of $\rho$ are given by:
\begin{equation}
\begin{array}{rclcrcl}
\rho^{[a]}_X & = &  \langle a, [0] \rangle\\[.3em]
\rho^{\X}_X & = &  \langle 0, [1] \rangle\\[.3em]
\rho^{+}_X(\langle a,x \rangle,\langle b,y \rangle) & = &
  \langle a+b, x+y \rangle\\[.3em]
\rho^{\times}_X(\langle a,x \rangle,\langle b,y \rangle) & = &
  \langle a \cdot b, (x \times [b]) + ((x \times (\X \times y)) + ([a] \times y)) \rangle
\end{array}
\end{equation}
As described at the end of section~\ref{ssec:dls-basic}, such a $\rho$ is a simple SOS specification, and it uniquely induces a distributive law $\lambda \colon TF \To FT$. This $\lambda$ is essentially the inductive extension of $\rho$ from terms of depth 1 to arbitrary terms.
Let $\E$ be given by the following axioms where $V=\{v,u,w\}$ and $a,b \in \bbR$ (see
Example~\ref{ex:equations} for an explanation of how this corresponds to a functor with two natural transformations):
\begin{equation}\label{eq:semiring-axioms}
\begin{array}{lll}
%(A,+,0)\text{ is commutative monoid}:
(v+u)+w = v+(u+w) \quad
& [0]+v = v
 & v+u = u+v
\\
%(A,\cdot,1)\text{ is monoid}:
(v\times u)\times w = v\times(u\times w)
& [1] \times v = v \qquad
& v \times u = u \times v
\\
%\text{Distribution:}
v\times (u + w) = (v \times u) + (v\times w) \qquad
%\text{Annihilation by 0:}
& [0] \times v = [0]
\\
{[a+b] = [a]+[b]} \quad
& [a \cdot b] = [a] \times [b] \quad
\end{array}
\end{equation}
$\E$ consists of the \emph{commutative} semiring axioms
together with axioms stating the inclusion of the underlying semiring of the reals.
We would like to apply Theorem~\ref{thm:lambda'-dlaw} to obtain a distributive
law $\lambda'$ for the quotient monad $\T'$ arising from $\T$ and $\E$.
We show that $\lambda$ preserves $\E$.
%Let $g\colon V \to FX$ be arbitrary and suppose $g(v)=\tup{a,x}, g(u)=\tup{b,y}, g(z)=\tup{c,z}$.
Let $\tup{a,x}, \tup{b,y}, \tup{c,z} \in FX$ for some set $X$.
First note that for $F = \bbR\times \Id$,
$\tup{r_1,t_1} \;\Rel{F}(\Eq{X})\;\tup{r_2,t_2}$ iff
$r_1=r_2$ and $t_1 \Eq{X} t_2$.
It is straightforward to check preservation of the axioms that only concern addition,
as well as of $[1] \times v = v$, $[0] \times v = [0]$ and $v \times u = u \times v$.
We show that $[a \cdot b] = [a] \times [b]$ is preserved:
\[ \begin{array}{rcl}
\lambda_X([a]\times [b])
&=&
\tup{a\cdot b , [0]\times [b] + [0] \times \X \times [0] + [a] \times [0]}
\\
&\Rel{F}(\Eq{X})&
\tup{a\cdot b , [0]}
=
\lambda_X([a \cdot b])
\end{array}\]
We check that $\lambda$ preserves the distribution axiom:
\[ \begin{array}{lcl}
\multicolumn{3}{l}{\lambda_X(\tup{a,x} \times (\tup{b,y}+ \tup{c,z}))}
\\
\quad & = & \tup{a \cdot (b+c), (x \times [b+c]) + (x \times X \times (y+z)) + [a]\times (y+z)}
\\
\quad & \Rel{F}(\Eq{X}) &
\tup{a \cdot (b+c), (x \times [b+c]) + (x \times X \times y) + (x \times X \times z) + \\
&& \qquad\qquad\qquad  \qquad\qquad\qquad \hfill([a]\times y) + ([a]\times z)}
\\
& \Rel{F}(\Eq{X}) &
\tup{(a\cdot c) + (b \cdot c),
    (x \times [b]) + (x \times X \times y) + ([a]\times y) + \\
&& \qquad\qquad\qquad  \qquad\qquad\qquad   (x \times[c]) + (x \times X \times z) + ([a]\times z)}
\\
& = &\\
\multicolumn{3}{l}{\lambda_X((\tup{a,x} \times \tup{b,y}) + (\tup{a,x} \times \tup{c,z}))}
\end{array}\]
Note that we used $[a+b] = [a]+[b]$.
Similarly, preservation of $\times$-associativity can be verified, and it uses
the axiom $[a \cdot b] = [a] \times [b]$.
We have thus shown that $\lambda$ preserves $\E$, and it follows, in particular, that
$\tup{\bbR^\omega,+,\times, [0],[1]}$ is a commutative semiring.
This was shown directly in \cite{Rut05:MSCS-stream-calc},
but the proof uses bisimulation-up-to as well as the fundamental theorem of stream calculus,
which cannot be added as an equation.
In our approach we construct a distributive law, and obtain not only this result but also the soundness of the bisimulation-up-to
technique~\cite{RBR13}, and the existence of unique solutions to corecursive equations
$\phi\colon X \to FT'X$ (see Section~\ref{ssec:corecursive-eqs}).

The derivative of the convolution product is usually
(cf.~\cite{Rut03:TCS-bde,Rut05:MSCS-stream-calc}) specified as:
\begin{equation}\label{eq:conv-gsos}
(\sig\times\tau)' = (\sig'\times\tau) + ([\sig(0)]\times\tau')
\end{equation}
which corresponds to a stream GSOS-rule
$\Sig(\Id\times \bbR \times \Id) \To \bbR \times T(-)$, and thus to a distributive law
over the cofree copointed functor.
However, with this definition, we could not show that the commutativity of $\times$
is preserved although all other axioms remain preserved.
Hence a given $\lambda$ does not necessarily satisfy all equations
that are valid on the final $F$-coalgebra.
\defiEnd
\end{example}

In the above example the monad under consideration is defined by operations and equations. In Example~\ref{ex:CFG-lambda} below we will see an example of a monad that has an independent definition, but where a presentation by operations and equations simplifies the construction of a distributive law considerably.

\begin{remark}
  The concrete proof method for preservation of equations bears a close
  resemblance to \emph{bisimulation up to congruence}~\cite{RBR13}, in that one must show
  that for every pair in the (image of the) equations, its derivatives are related by the least
  \emph{congruence} $\Eq{X}$ instead of just the equivalence relation induced by the equations.
\end{remark}

\begin{example}\label{exle:not-preserve}
As we discussed at the end of Example~\ref{ex:stream-calc} (regarding the definition of the convolution product), it is not always possible to show that a given $\lambda$ preserves all equations that hold in the final coalgebra. Now we give another concrete example of this fact. This
example again concerns stream systems, i.e., coalgebras for the functor $FX = \bbR \times X$.
We define the constant stream of zeros by three different constants $n_1, n_2$ and $n_3$ by the following behavioural differential equations:
\[
\begin{array}{rclcrclcrclcrclcrclcrcl}
n_1(0) &=& 0, &  n_1' &=& n_1  & \quad &
{n}_2(0) &=& 0, & n_2' &=& n_3  & \quad &
{n}_3(0) &=& 0, & n_3' &=& n_3
\end{array}
\]
The corresponding signature functor is thus $\Sigma X = 1 + 1 + 1$, and the above specification gives rise to
a distributive law $\lambda \colon TF \To FT$ where $T$ is (the functorial component of) the free monad over $\Sigma$. Now consider
the equation $n_1 = n_2$; this clearly holds when interpreted in the final coalgebra. However, this
equation is not preserved by $\lambda$. To see this, notice that $\lambda(n_1) = \langle 0, n_1 \rangle$
and $\lambda(n_2) = \langle 0, n_3 \rangle$, but $n_1 \not \Eq{X} n_3$, so $\lambda(n_1)$ and $\lambda(n_2)$
are not related by $\Rel{F}(\Eq{X})$.
\defiEnd
\end{example}

\subsection{Distributive Laws over Copointed Functors}

We now show that our main results hold as well for
distributive laws of monads over \emph{copointed} functors.
This extends our method to deal with operations specified in the
abstract GSOS format, such as language concatenation.

\newcommand{\propdlcopointed}{
Theorem~\ref{thm:lambda'-dlaw} and Corollary~\ref{cor:pres-dl} hold as well for any distributive law of a monad over a copointed functor.
}
\begin{proposition}\label{prop:dl-copointed}
  \propdlcopointed
\end{proposition}

\begin{proof}
Let $\tup{H, \epsilon}$ be a copointed functor and $\lambda \colon TH \Rightarrow HT$ a distributive law
of $\T$ over $\tup{H, \epsilon}$. Suppose $\lambda$ preserves equations $E$. By Theorem~\ref{thm:lambda'-dlaw} then there is a distributive law $\lambda'$
of $\T'$ over $H$ such that $q \colon T \To T'$ is a morphism of distributive laws. In order to show that
$\lambda'$ is a distributive law of $\T'$ over $\tup{H, \epsilon}$ we only need to prove
that $\lambda'$ satisfies the additional axiom, i.e., that the right crescent in the following diagram commutes:
$$
  \xymatrix{
    THX \ar[r]^{q_{HX}} \ar[d]^{\lambda_X} \ar@/_2pc/[dd]_{T\epsilon_X} & T'HX \ar[d]_{\lambda'_X} \ar@/^2pc/[dd]^{T'\epsilon_X} \\
    HTX \ar[r]^{Hq_{X}} \ar[d]^{\epsilon_{TX}}  & HT'X \ar[d]_{\epsilon_{T'X}} \\
    TX \ar[r]^{q_{X}}                           & T'X
  }
$$
The outermost part commutes by naturality of $q$, the upper square
commutes since $\lambda$ is a morphism of distributive laws, and the
lower square commutes by naturality of $\epsilon$, and the left
crescent commutes by the fact that $\lambda$ is a distributive law of
$\T$ over $\tup{H,\epsilon}$. Consequently we have
$\epsilon_{T'X} \circ \lambda'_X \circ q_{HX} = T'\epsilon_X \circ
q_{HX}$,
and since $q_{HX}$ is an epi we obtain
$\epsilon_{T'X} \circ \lambda'_X = T'\epsilon_X$ as desired.

For Corollary~\ref{cor:pres-dl} one needs to add to its proof a check that the distributive law satisfies the additional
axiom as well, which is again rather easy to do.
\end{proof}

\begin{example}[Context-free languages]\label{ex:CFG-lambda}
A context free grammar (in Greibach normal form) consists of a finite set $A$ of terminal
symbols, a (finite) set $X$ of non-terminal symbols, and a map
$\langle o,t \rangle \colon X \to 2 \times \Pow_\omega(X^*)^A$,
i.e., it is a coalgebra for the behavior functor $F(X) = 2 \times X^A$ composed
with the idempotent semiring monad $\Pow_\omega((-)^*)$ from Example~\ref{ex:ism}.
Intuitively, $o(x) = 1$ means
that the variable $x$ can generate the empty word, whereas $w \in t(x)(a)$ if and
only if $x$ can generate $aw$, cf.~\cite{WBR:CF-CALCO}.

It is a rather difficult task to describe concretely a distributive law of $T'=\Pow_\omega((-)^*)$
over $F$ (or $\Id \times F$) defining the
sum $+$ and sequential composition $\cdot$ of context-free grammars. More conveniently, since
we have seen in Example~\ref{ex:ism} that the monad $\Pow_\omega((-)^*)$
can be presented
by the operations and axioms of idempotent semirings, we proceed by
defining a distributive law $\lambda$ of the free monad $\T_\Sig$ generated by the semiring signature functor
$\Sig(X)= 1 + 1 + (X \times X)+ (X \times X)$ over the cofree copointed functor
$\tup{\Id\times F, \pi_1}$, and show that $\lambda$ preserves the semiring
axioms.
We define $\lambda$ as the distributive law that corresponds to the natural transformation
$\rho\colon \Sig(\Id \times F) \To FT$ whose components are given by:
\begin{equation}
\begin{array}{rclcrcl}
\rho^{\cns{0}}_X & = &  \langle 0, a \mapsto \emptyset \rangle\\[.3em]
\rho^{\cns{1}}_X & = &  \langle 1, a \mapsto \emptyset \rangle\\[.3em]
\rho^{+}_X(\langle x,o,f \rangle,\langle y,p,g \rangle) & = &
  \langle \text{max}\{ o, p \}, a \mapsto f(a) + g(a) \rangle\\[.3em]
\rho^{\cdot}_X(\langle x,o,f \rangle,\langle y,p,g \rangle) & = &
  \left \langle \text{min}\{ o, p \}, a \mapsto  \begin{cases}
                        f(a) \cdot y        & \mbox{if $p = 0$}\\[.3em]
                        f(a) \cdot y + g(a) & \mbox{if $p=1$}
                        \end{cases}
  \right \rangle
\end{array}
\end{equation}
We proceed to show that $\lambda$ preserves the defining equations of idempotent semirings.
We treat here only the case of distributivity, i.e., $u\cdot (v + w) = u \cdot v + u \cdot w$. To this end, let $X$ be arbitrary and suppose $\tup{x,o,d}, \tup{y,p,e}, \tup{z,q,f} \in X \times FX$. Notice that either $o=0$ or $o=1$; we treat both cases separately:
$$
\begin{array}{lcl}
  \multicolumn{3}{l}{\lambda(\tup{x, 0, d} \cdot (\tup{y, p, e} + \tup{z, q, f}))} \\
  \quad &=& (x \cdot (y + z), 0, a \mapsto d(a) \cdot (y + z)) \\
    &\Rel{F}(\Eq{X})& (x \cdot y + x \cdot z, 0, a \mapsto d(a) \cdot y + d(a) \cdot z) \\
    &=& \lambda(\tup{x, 0, d} \cdot \tup{y, p, e} + \tup{x, 0, d} \cdot \tup{z, q, f})\\
    & & \\
  \multicolumn{3}{l}{\lambda(\tup{x, 1, d} \cdot (\tup{y, p, e} + \tup{z, q, f}))}  \\
  \quad  &=& (x \cdot (y + z), p+q, a \mapsto d(a) \cdot (y + z) + (e(a) + f(a))) \\
    &\Rel{F}(\Eq{X})& (x \cdot y + x \cdot z, p+q, a \mapsto (d(a) \cdot y + d(a) \cdot z) + (e(a) + f(a))) \\
    &\Rel{F}(\Eq{X})& (x \cdot y + x \cdot z, p+q, a \mapsto (d(a) \cdot y + e(a)) + (d(a) \cdot z + f(a))) \\
    &=& \lambda(\tup{x, 1, d} \cdot \tup{y, p, e} + \tup{x, 1, d} \cdot \tup{z, q, f})\,.
\end{array}
$$
In a similar way one can show that $\lambda$ preserves the other idempotent semiring equations. Thus, from
Proposition~\ref{prop:dl-copointed} and Corollary~\ref{cor:pres-dl} we obtain a distributive law $\kappa$ of $\Pow_\omega((-)^*)$ over $\Id \times F$ such
that $i \circ q\colon \lambda\To\kappa$  is a morphism of distributive laws, i.e., $\kappa$ is presented
by $\lambda$ and the equations of idempotent semirings.
\defiEnd
\end{example}

\subsection{Distributive Laws over Comonads}

A further type of distributive law, which generalizes all of the above, is that of a distributive law of a monad over a comonad.
These arise from GSOS laws as well as from \emph{coGSOS} laws, which allow to model operational rules which involve look-ahead
in the premises. We refer to~\cite{Klin11} for technical details and an example of a coGSOS format on streams. In this
subsection, we prove for future reference  that when constructing the quotient distributive law as above for
a distributive law over a comonad, the axioms are preserved, i.e., the quotient is again a distributive law over the comonad.

\newcommand{\propdlcomonad}{
Theorem~\ref{thm:lambda'-dlaw} and Corollary~\ref{cor:pres-dl} hold as well for any distributive law of a monad over a comonad.
}
\begin{proposition}\label{prop:dl-comonad}
  \propdlcomonad
\end{proposition}

\begin{proof}
Let $\tup{D, \epsilon, \delta}$ be a comonad and $\lambda \colon TD \Rightarrow DT$ a distributive law of the monad $\tup{T,\eta,\mu}$
over the comonad $\tup{D,\epsilon,\delta}$.
Suppose $\lambda$ preserves equations $\E$. By Proposition~\ref{prop:dl-copointed} there is a distributive law $\lambda'$
of $\T'$ over the copointed functor $\tup{D,\epsilon}$. To show that $\lambda'$ is a distributive law
over the comonad $\tup{D, \epsilon, \delta}$, we need to check that the corresponding axiom holds.
$$
  \xymatrix{
    TD \ar[d]^{T\delta} \ar@/_3.5pc/[ddd]_{q_{D}} \ar[rr]^\lambda & & DT \ar[d]^{\delta_T} \ar@/^3.5pc/[ddd]^{Dq} \\
    TDD \ar[r]^{\lambda_D} \ar[d]^{q_{DD}} & DTD \ar[r]^{D\lambda} \ar[d]^{Dq_{D}} & DDT \ar[d]^-{DDq} \\
    T'DD \ar[r]^{\lambda'_D} & DT'D \ar[r]^{D\lambda'} & DDT' \\
    T'D \ar[u]_{T'\delta} \ar[rr] ^{\lambda'}& & DT' \ar[u]_{\delta_{T'}}
  }
$$
The outer square as well as the two small inner squares commute by the 
fact that $q$ is a morphism of distributive laws. The outer crescents 
commute since $q$ and $\delta$ are natural.
The upper rectangle commutes by the assumption that $\lambda$ is a distributive law
over the comonad. Checking that the lower rectangle
commutes, which is what we need to prove, is now an easy diagram chase, using that $q_{D}$ is epic.
\end{proof}

%------------------------------------------------------------------

\section{Morphisms and Solutions}
\label{sec:quotients-of-bialgebras}

In this section, we show that morphisms of distributive laws
commute with solving corecursive equations.
In the case of monads with equations, this means that
first solving equations $\phi$ with respect to $\T$ and then forming the quotient of
the solution bialgebra
is the same as first forming the quotient of $\T$ and solving with respect to the
quotient monad $\T'$.

We first describe some functors that link the relevant categories
of bialgebras and corecursive equations.
Throughout this Section, we let $\T=\tup{T,\eta,\mu}$
and $\K = \tup{K,\theta,\nu}$ be monads; and
$\lambda\colon TF \To FT$ and $\kappa\colon KF \To FK$
be distributive laws of $\T$ and $\K$ over $F$, respectively.

If $\tau\colon \lambda\To\kappa$ is a morphism of distributive laws,
then precomposing with $\tau$ yields a functor:
\begin{equation}\label{eq:def-I}
\begin{array}{rrcl}
I \colon & \Bialg(\kappa) & \to & \Bialg(\lambda)\\
&  \xymatrix@1{KX \ar[r]^-{\alpha} & X \ar[r]^-{\beta} & FX}
 & \quad\mapsto\quad
 & \xymatrix@1{TX \ar[r]^-{\alpha\circ\tau_X} & X \ar[r]^-{\beta} & FX}
\end{array}
\end{equation}

 It follows from the naturality of $\tau$ and
  $F\tau\circ\lambda=\kappa\circ \tau F$ that
  $I$ takes a $\kappa$-bialgebra to a $\lambda$-bialgebra.
Similarly, postcomposing with $F\tau$ yields a functor between corecursive equations:
\begin{equation}\label{eq:def-Q}
\begin{array}{rrcl}
Q \colon & \Coalg(FT) & \to & \Coalg(FK)\\[.2em]
& \phi\colon X \to FTX & \quad \mapsto \quad & F\tau_X\circ\phi\colon X \to FKX
\end{array}
\end{equation}
Recall from Section~\ref{ssec:corecursive-eqs}, that
given a distributive law $\lambda\colon TF \To FT$,
the solutions of a corecursive equation $\phi\colon X \to FTX$
are characterised by morphisms from the $\lambda$-bialgebra
$\tup{TX,\mu_X,\ext{\phi}{\lambda}}$
whose $F$-coalgebra structure given by
\begin{equation}
\ext{\phi}{\lambda} = F\mu_X \circ \lambda_{TX} \circ T\phi
\end{equation}
This yields a functor (see, e.g.,\cite[Lem.~5.4.11]{BJ:Coalg-book-draft}):
\begin{equation}\label{eq:def-G}
\begin{array}{rrcl}
G_\lambda \colon & \Coalg(FT) & \to & \Bialg(\lambda)\\[.2em]
& \tup{X,\phi} & \quad\mapsto\quad  & \tup{TX,\mu_X,\ext{\phi}{\lambda}}
\end{array}
\end{equation}
We can go in the opposite direction by using the monad unit,
\begin{equation}\label{eq:def-V}
\begin{array}{rrcl}
V_\eta \colon & \Bialg(\lambda) & \to & \Coalg(FT) \\[.2em]
&  \tup{X,\alpha,\beta} & \quad \mapsto\quad  & \tup{X, F\eta_X\circ\beta}
\end{array}
\end{equation}
which decomposes into the functor $U\colon \Bialg(\lambda)\to \Coalg(F)$
that forgets algebra structure, and
\begin{equation}\label{eq:def-J}
\begin{array}{rrcl}
J_\eta \colon & \Coalg(F) & \to & \Coalg(FT) \\[.2em]
&  \tup{X,\beta} & \quad\mapsto\quad  & \tup{X, F\eta_X\circ\beta}
\end{array}
\end{equation}
The following diagram summarises the situation:
\begin{equation}\label{eq:IQGV-diagram}
\xymatrix@R=1.5em{
\Bialg(\lambda) \ar@/^1.2pc/[r]^-{V_\eta} \ar@/^5pc/[drr]_-{U} \ar@{}[r]|{}
& \Coalg(FT) \ar@/^1.2pc/[l]^-{G_\lambda}  \ar[dd]^-{Q} &
\\
&& \Coalg(F)
\ar[ul]_-{J_\eta}\ar[dl]^-{J_\theta} \\
\Bialg(\kappa) \ar@/^1.2pc/[r]^-{V_\theta} \ar[uu]_-{I} \ar@{}[r]|{}
\ar@/_5pc/[urr]^-{U}
& \Coalg(FK) \ar@/^1.2pc/^-{G_\kappa}[l]
}
\end{equation}
We mention that $QV_\eta I = V_\theta$ since $\tau$ is compatible with the
units of $\T$ and $\K$. %\blue{(More formal proof?)}

Morphisms of distributive laws are defined to be monad maps, and hence
respect the algebraic structure. The next proposition shows that,
as one might expect, they also respect the coalgebraic structure,
and hence morphisms of distributive laws induce
morphisms between bialgebras.

\newcommand{\proptaubialgebramorphism}{
If $\tau\colon \lambda \To \kappa$ is a morphism of distributive laws,
then for all
  $\phi:X\to FTX$  we have that $\tau_{X}$ is a $\lambda$-bialgebra morphism
$
\tau_X\colon
G_\lambda(\phi) \to I G_\kappa Q(\phi)
%\tup{TX,\mu_X,\ext{\phi}{\lambda}}\to I\tup{KX,\nu_X,\ext{(Q\phi)}{\kappa}}
$
or, equivalently, an $F$-coalgebra morphism
$\tau_X \colon \tup{TX,\ext{\phi}{\lambda}}\to \tup{KX,\ext{(Q\phi)}{\kappa}}$.
}
\begin{proposition}\label{prop:tau-bialgebra-morphism}
\proptaubialgebramorphism
\end{proposition}

\begin{proof}
We show that
$\tau_X \colon \tup{TX,\ext{\phi}{\lambda}}\to \tup{KX,\ext{(Q\phi)}{\kappa}}$
is an $F$-coalgebra morphism:
\[\xymatrix@C=3.5em{
TX \ar[d]_-{T\phi} \ar[r]^-{\tau_X} \ar@{}[dr]|{(\text{nat.}\tau)}
& KX \ar[d]^-{K\phi} \ar@/^1.5pc/[dr]^-{KQ\phi} \ar@{}[dr]|{(\text{def.}Q\phi)}
\\
TFTX \ar[d]_-{\lambda_{TX}} \ar[r]^-{\tau_{FTX}} \ar@{}[dr]|{\eqref{eq:morphism-dl-square}}
& KFTX \ar[d]^-{\kappa_{TX}} \ar[r]^-{KF\tau_X} \ar@{}[dr]|{(\text{nat.}\lambda)}
& KFKX \ar[d]^-{\kappa_{KX}}
\\
FT^2X \ar[d]_-{F\mu_X} \ar[r]^-{F\tau_{TX}}  \ar@{}[drr]|{F(\tau~\text{monad morph.})}
& FKTX  \ar[r]^-{FK\tau_X}
& FKKX \ar[d]^-{F\nu_X}
\\
FTX \ar[rr]_-{F\tau_X} && FKX
}\]
\end{proof}

It follows that the unique $\lambda$-bialgebra morphism
$g \colon \tup{TX,\mu_X,\ext{\phi}{\lambda}} \to \tup{Z,\alpha,\zeta}$
into the final $\lambda$-bialgebra $\tup{Z,\alpha,\zeta}$
factors as $g = g'\circ\tau_X$, where $g'$ is the final $\lambda$-bialgebra
morphism from $IG_\kappa Q(\phi)$,
as shown here:
\begin{equation}\label{eq:solution-quotient}
\xymatrix@C=1em{
& T^2X \ar[rr]^-{T\tau_X} \ar[d]^-{\mu_X}
&& TKX \ar[rr]^-{Tg'}  \ar[d]^-{\nu_X\circ \tau_{KX}}
&& TZ \ar[d]^-{\alpha}
\\
X \ar[r]^-{\eta_X} \ar[dr]_-{\phi}
& TX \ar[rr]^-{\tau_X} \ar[d]^-{\ext{\phi}{\lambda}}
&& KX \ar[rr]^-{g'}  \ar[d]^-{\ext{(Q\phi)}{\kappa}}
&& Z \ar[d]^-{\zeta}
\\
& FTX \ar[rr]^-{F\tau_X}
&& FKX \ar[rr]^-{Fg'}
&& FZ
}\end{equation}
Hence by Proposition~\ref{prop:solution},
every solution of $\phi$ in the final $\lambda$-bialgebra
yields a solution of $Q\phi$, and vice versa.

When $\tau\colon \lambda\To\kappa$ arises from
a set of preserved equations $\E$ as in Section~\ref{sec:quotients-of-dls}
(with $\kappa=\lambda'$),
then Proposition~\ref{prop:tau-bialgebra-morphism}
says that $IG_{\kappa} Q(\phi)$ is a quotient of the
``free'' $\lambda$-bialgebra $\tup{TX,\mu_X,\ext{\phi}{\lambda}}$,
and in particular, the congruence $\Eq{X}$ is an $F$-behavioural equivalence.
In this case, $Q\phi$ is the corecursive equation
obtained by reading the right-hand side of $\phi$ modulo equations in $E$.
In other words,
forming the quotient of the solution of the equation $\phi$ is the same
as solving the quotiented equation $Q\phi$.

\begin{example}\label{ex:CFG-lambda-2}
Recall from Example~\ref{ex:CFG-lambda} that $i \circ q \colon T \Rightarrow \Pow_\omega(X^*)$ is a morphism
of distributive laws. By Proposition~\ref{prop:tau-bialgebra-morphism} we have the following commuting
diagram for any corecursive equation $\phi \colon X \rightarrow 2 \times (TX)^A$:
\begin{equation}\label{eq:solution-quotient-CFG}
\xymatrix@C=1.8em{
X \ar[r]^-{\eta_X} \ar[dr]_-{\phi}
& TX \ar[rr]^-{(i \circ q)_X} \ar[d]^-{\ext{\phi}{\lambda}}
&& \Pow_\omega(X^*) \ar[rr]  \ar[d]^-{\ext{(Q\phi)}{\kappa}}
&& \Pow(A^*) \ar[d]^-{\zeta}
\\
& 2 \times (TX)^A \ar[rr]^-{~~\id \times ((i \circ q)_X)^A}
&& 2 \times (\Pow_\omega(X^*))^A \ar[rr]
&& 2 \times \Pow(A^*)^A
}\end{equation}
Notice that a context-free grammar $\langle o,t \rangle \colon X \rightarrow 2 \times \Pow_\omega(X^*)^A$ can be represented
by a $\phi\colon X \to 2 \times (TX)^A$ such that $Q\phi = \langle o, t\rangle$, since $i \circ q$ is surjective. This
gives the expected correspondence between two of the three different coalgebraic approaches to context-free languages introduced in~\cite{WBR:CF-CALCO}
(the third approach is about fixed-point expressions and as such is outside the scope of this paper).
\defiEnd
\end{example}

Similarly, the algebraic structure induced by $\lambda$ on the
final $F$-coalgebra factors uniquely through the algebraic structure
induced by $\kappa$.
\newcommand{\propalgsfinal}{
  Let $\tau\colon \lambda \To \kappa$ be a morphism of distributive laws, and
  let $\alpha \colon TZ \rightarrow Z$ and $\alpha' \colon KZ \rightarrow Z$
  be the algebras induced by $\lambda$ and $\kappa$ respectively on the final coalgebra $\tup{Z, \zeta}$.
  Then $\alpha = \alpha' \circ \tau_Z$.
}
\begin{proposition}\label{prop:algs-final}
  \propalgsfinal
\end{proposition}
\begin{proof}
%  Recall from Section~\ref{sec:dls} that $\alpha$ and $\alpha'$ arise by finality, being the unique
%  coalgebra homomorphisms from $\tup{TZ, \lambda_Z \circ T \zeta}$ and $\tup{KZ, \kappa_Z \circ K \zeta}$ respectively
%  into the final coalgebra $\tup{Z, \zeta}$.
  Consider the following diagram:
  $$
  \xymatrix{
    TZ \ar[r]^{\tau_Z} \ar[d]^{T\zeta}         & KZ \ar[r]^{\alpha'} \ar[d]^{K\zeta}   & Z \ar[dd]^{\zeta}   & TZ \ar[d]^{T\zeta} \ar[l]_{\alpha}\\
    TFZ \ar[r]^{\tau_{FZ}} \ar[d]^{\lambda_Z}  & KFZ \ar[d]^{\kappa_Z}             &                     & TFZ \ar[d]^{\lambda_Z} \\
    FTZ \ar[r]^{F\tau_Z}                       & FKZ \ar[r]^{F\kappa}              & FZ                  & FTZ \ar[l]_{F\alpha}
  }
  $$
  The upper left square commutes by naturality of $\tau$, whereas the lower left square commutes since $\tau$ is a morphism of distributive
  laws. The two rectangles commute by definition of $\alpha$ and $\alpha'$ (see Section~\ref{sec:dls}). Thus $\alpha' \circ \tau_Z$ and $\alpha$ are both
  coalgebra homomorphisms from $\tup{TZ, \lambda_Z \circ T\zeta}$ to $\tup{Z, \zeta}$ and consequently $\alpha' \circ \tau_Z = \alpha$
  by finality.
\end{proof}

\begin{example}\label{ex:CFG-lambda-3}
Continuing Example~\ref{ex:CFG-lambda-2}, it follows from Proposition~\ref{prop:algs-final} that
the algebra $\alpha \colon T\Pow(A^*) \rightarrow \Pow(A^*)$ induced by the distributive law for
$\T$ can be decomposed as $i \circ q \circ \alpha'$, where $\alpha'$ is the algebra on $\Pow(A^*)$
induced by the distributive law for $\Pow_\omega(\Id^*)$. It can be shown by induction that
$\alpha$ is the algebra on languages given by union and concatenation product. Now
$\alpha' \colon \Pow_\omega(\Pow(A^*)^*) \to \Pow(A^*)$ can be given by selecting a representative
term and applying $\alpha$, and it follows that
  $\alpha'(\mathcal{L}) = \bigcup_{L_1\cdots L_n \in \mathcal{L}} \{ w_1\cdots w_n \mid w_i \in L_i \}$.
\defiEnd
\end{example}

%------------------------------------------------------------------

\section{Discussion and Conclusion}
\label{sec:conc}

% Summary
We have presented a preservation condition
that is necessary and sufficient for the existence of a distributive law $\lambda'$
for a monad with equations given a distributive law $\lambda$
for the underlying free monad.
This condition consists of checking that the base equations are preserved
by $\lambda$.
Example \ref{ex:CFG-lambda} shows that presenting a monad by operations and equations and then checking that $\lambda$ preserves the equations can be much easier than describing and verifying the
distributive law requirements directly.
We demonstrated our method by applying it to obtain distributive laws
for stream calculus over commutative semirings, and
for context-free grammars which use the monad of idempotent semirings.

% Related work:
In~\cite{Watanabe:cmcs2002} the notion of morphisms of distributive laws is studied as a general approach
to translations between operational semantics. In this paper we investigate in detail the case of
quotients of distributive laws.
Distributive laws for monad quotients and equations are also studied in
\cite{LPW2004:cat-sos,Manes:MonadComp}.
The setting and motivation of \cite{Manes:MonadComp} is different
as they study distributive laws of one monad over another with the aim to
compose these monads. We study distributive laws of a monad
over a plain functor, a copointed functor or a comonad.
The approach in \cite{LPW2004:cat-sos} differs from ours in that
the desired distributive law is contingent on two given distributive laws and
the existence of the coequaliser (in the category of monads) which encodes
equations.
We have given a more direct analysis for monads in $\Set$ and a practical proof principle,  which covers
many known examples. We leave as future work to find out precisely how their Theorem~31
relates to our Theorem~\ref{thm:lambda'-dlaw}. 
In~\cite{ASP13} effects with equations are added to the syntax generated by a free monad $T$,
using as semantic domain a suitable final $B$-coalgebra in the Kleisli category of $T$ (assumed to 
be enriched over $\omega$-complete pointed partial-orders). To prove adequacy of the semantics with 
respect to a given operational model, the authors use a result similar to our Theorem~\ref{thm:lambda'-dlaw}. 
Their result, however, is limited to coalgebras for the functor $BX = V + X$. Moreover, since we 
work in Eilenberg-Moore categories of algebras rather than Kleisli categories of free algebras,
we do not need to require the monad (and the quotient map) to be strong.

% Future work:
While in this work we have focused on adding equations which already hold in the final bialgebra, it is
often useful to use equations to \emph{induce} behaviour, next to a behavioural specification
in terms of a distributive law. In process theory this idea is captured by the notion of structural congruences~\cite{MousaviR05}.
At the more general level of distributive laws there is work on adding recursive equations~\cite{Klin04}.
A study of structural congruences for distributive laws on free monads was given recently in~\cite{RB14}. While that work
focuses only on free monads, we believe that it can possibly be combined with the present work to give a more general account of
equations and structural congruences for different monads.

In the case of GSOS on labelled transition systems, proving equations to hold at the level of a specification
was considered in~\cite{AcetoCI12}, based on \emph{rule-matching bisimulations}, a refinement of De Simone's notion of
FH-bisimulation. Rule-matching bisimulations are based on the syntactic notion of \emph{ruloids}, while our technique
is based on preservation of equations at the level of distributive laws. It is currently not clear what the precise
relation between these two approaches is; one difference
is that preserving equations naturally incorporates reasoning up to congruence.

More technically, it remains an open problem whether a converse of 
Proposition~\ref{prop:tau-bialgebra-morphism} holds. For the special 
case of $\C=\Set$, such a converse has been proved by Joost Winter (in 
personal communication). We intend to investigate the general case in 
future work.

%\bibliographystyle{plain}
%\bibliography{lmcs}

%\appendix
%\section{}

%  Doctors recommend to take out the appendix if pain starts to get
%  life threatening.

\end{document}